\documentclass{article}
\usepackage[T1]{fontenc}
\usepackage{graphicx,xcolor}
\usepackage{epsfig,amsmath,amssymb,amsthm,cite,algorithm,algorithmic,dcolumn,tikz}
\usetikzlibrary{arrows,snakes,backgrounds,positioning}

\theoremstyle{plain}
\newtheorem{theorem}{Theorem}
\newtheorem{proposition}[theorem]{Proposition}
\newtheorem{lemma}[theorem]{Lemma}
\newtheorem{corollary}[theorem]{Corollary}
\theoremstyle{definition}

\newtheorem{example}[theorem]{Example}
\newtheorem{problemenvironment}[theorem]{Problem}

\newcommand{\problemdef}[3]{
\begin{problemenvironment}[#1]
{\ }
\begin{itemize}
\item[] \textbf{Instance:} {#2}
\item[] \textbf{Question:} {#3}
\end{itemize}
\end{problemenvironment}
}

\newcommand{\lrangle}[1]{\langle #1 \rangle}
\newcommand{\ceil}[1]{\left\lceil #1 \right\rceil}
\newcommand{\floor}[1]{\left\lfloor #1 \right\rfloor}
\newcommand{\ul}[1]{\underline{#1}}

\newcommand{\prg}[1]{\noindent\tbf{{#1}}}

\newcommand{\commentout}[1]{}

\newcommand{\mbb}[1]{\mathbb{#1}}
\newcommand{\mrm}[1]{\mathrm{#1}}
\newcommand{\msf}[1]{\mathsf{#1}}
\newcommand{\tbf}[1]{\textbf{#1}}

\newcommand{\vecp}[1]{\vec{#1}^{\,\prime}}
\newcommand{\vecpp}[1]{\vec{#1}^{\,\prime\prime}}

\newcommand{\OPT}{\mathsf{ts}}
\newcommand{\POPT}{\mathsf{rt}}
\newcommand{\SOL}{\mathsf{TS}}
\newcommand{\PSOL}{\mathsf{RT}}
\newcommand{\PathAlg}{\mathsf{AP}}
\newcommand{\OES}{\textsf{\OE}}
\newcommand{\SAT}{Sep-SAT}

\newcommand{\tabResults}[1]{}
\newcommand{\figDMTSP}[1]{}
\newcommand{\figSATPTSPA}[1]{}
\newcommand{\figSATPTSPB}[1]{}

\title{The Time Complexity of Permutation Routing via Matching, Token Swapping and a Variant}

\author{%
Jun Kawahara
\\
Graduate School of Information Science, \\ Nara Institute of Science and Technology, Japan
 \and
Toshiki Saitoh
\\
Faculty of Computer Science and Systems Engineering,
\\ Kyushu Institute of Technology, Japan
 \and
Ryo Yoshinaka
\\
Graduate School of Information Sciences, Tohoku University, Japan
} 

\begin{document}
\maketitle              

\begin{abstract}
The problems of Permutation Routing via Matching and Token Swapping are reconfiguration problems on graphs. 
This paper is concerned with the complexity of those problems and a colored variant.
For a given graph where each vertex has a unique token on it, those problems require to find a shortest way to modify a token placement into another by swapping tokens on adjacent vertices.
While all pairs of tokens on a matching can be exchanged at once in Permutation Routing via Matching, 
Token Swapping allows only one pair of tokens can be swapped.
In the colored version, vertices and tokens are colored and the goal is to relocate tokens so that each vertex has a token of the same color.
We investigate the time complexity of several restricted cases of those problems and show when those problems become tractable and remain intractable.
\end{abstract}

\section{Introduction}
Alon et al.~\cite{AlonCG94} have proposed a problem called \emph{Permutation Routing via Matching} as a simple variant of routing problems.%
\footnote{In the preliminary version~\cite{KawaharaSY17} of this paper, this old problem was called Parallel Token Swapping due to the ignorance of the authors.}
Suppose that we have a simple graph where each vertex is assigned a token.
Each token is labeled with its unique goal vertex, which may be different from where the token is currently placed.
We want to relocate every misplaced token to its goal vertex.
What we can do in one step is to pick a matching and swap the two tokens on the ends of each edge in the matching.
The problem is to decide how many steps are needed to realize the goal token placement.
The bottom half of Figure~\ref{fig:TSP} illustrates a problem instance and a solution.
The graph has 4 vertices $1,2,3,4$ and 4 edges $\{1,2\},\{1,3\},\{2,4\},\{3,4\}$.
Each token $i$ is initially put on the vertex $5-i$.
By swapping the tokens on the edges in the matchings $\{\{1,2\},\{3,4\}\}$ and $\{\{1,3\},\{2,4\}\}$ in this order, we achieve the goal.
The original paper of Alon et al.~\cite{AlonCG94} and following papers are mostly interested in the maximum number of steps, denoted $\POPT(G)$, needed to realize the goal configuration from any initial configuration for an input graph $G$.
For example, Alon et al.\ have shown $\POPT(K_n) = 2$ for complete graphs $K_n$, Zhang~\cite{Zhang99} has shown $\POPT(T)=3n/2+(\log n)$ for trees $T$ of $n$ vertices, and Li et al.~\cite{LiLY10} have shown $\POPT(K_{m,n}) \in \floor{3m/2n}+O(1)$ for bipartite graphs $K_{m,n}$ with $m \ge n$ and $\POPT(C_n)=n-1$ for $n \ge 3$ for cycles $C_n$.
This paper is concerned with the problem where an initial configuration also constitutes an input and discusses its computational complexity.
We will show the following results, which were independently obtained by Banerjee and Richards~\cite{BanerjeeR17}.
\begin{itemize}
	\item Permutation routing via matching is NP-complete even to decide whether an instance admits a 3 step solution (Theorem~\ref{thm:3PTSPNPhard}).
	\item To decide whether 2 step solution exists can be answered in polynomial-time (Theorem~\ref{thm:2PTSPPoly}).
\end{itemize}
In addition, we present a polynomial-time algorithm that approximately solves Permutation Routing via Matching on paths whose output is at most one larger than that of the exact answer (Theorem~\ref{thm:PTSPPath}).
\begin{figure}[t]
\begin{center}
\newcommand{\sq}[4]{
\node (1) at (0,5) [vertex] {$1$};
\node (2) at (0.9,5) [vertex] {$2$};
\node (3) at (0,4.1) [vertex] {$3$};
\node (4) at (0.9,4.1) [vertex] {$4$};
\draw (1) -- (2) -- (4) -- (3) -- (1);
\node at (-0.3,5.3) [token] {#1};
\node at (1.2,5.3) [token] {#2};
\node at (-0.3,3.8) [token] {#3};
\node at (1.2,3.8) [token] {#4};
}
\newcommand{\sr}[5]{
\sq{#1}{#2}{#3}{#4}
\draw[->,semithick] (1.4,4.55) -- node[auto] {#5} (2.2,4.55);
\pgftransformxshift{3cm}
}
\tikzstyle{vertex}=[circle,draw,inner sep=0pt,minimum size=4mm]
\tikzstyle{token}=[rectangle,draw,fill=white,inner sep=0pt,minimum size=3mm]
\begin{tikzpicture}[scale=0.8]\small
\sr{$4$}{$3$}{$2$}{$1$}{$\{3,4\}$}
\sr{$4$}{$3$}{$1$}{$2$}{$\{1,3\}$}
\sr{$1$}{$3$}{$4$}{$2$}{$\{2,4\}$}
\sr{$1$}{$2$}{$4$}{$3$}{$\{3,4\}$}
\sq{$1$}{$2$}{$3$}{$4$}
\pgftransformxshift{-6cm}
\pgftransformyshift{-2.5cm}
\sq{$3$}{$4$}{$1$}{$2$}
\pgftransformxshift{-6cm}
\pgftransformyshift{2.5cm}
\draw[->,ultra thick] (0.45,3.6) -- (0.45,2.1) -- node[auto] {$\{\{1,2\},\{3,4\}\}$} (5.3,2.1);
\draw[<-,ultra thick] (12.45,3.6) -- (12.45,2.1) -- node[auto,swap] {$\{\{1,3\},\{2,4\}\}$} (7.5,2.1);
\end{tikzpicture}
\end{center}
\caption{\label{fig:TSP}%
Vertices and tokens are shown by circles and squares, respectively.
Optimal solutions for Token Swapping and Permutation Routing via Matching are shown by small and big arrows, respectively.}
\end{figure}
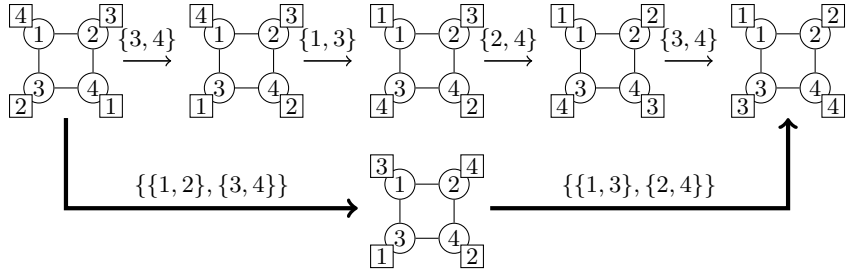

\emph{Token swapping}, introduced by Yamanaka et al.~\cite{YamanakaDIKKOSSUU14}, can be seen as permutation routing via ``edges''.
In this setting we can swap only two tokens on an edge at each step.
Figure~\ref{fig:TSP} shows that we require 4 steps in Token Swapping to realize the goal configuration, while 2 steps are enough in Routing via Matching.
Yamanaka et al.\ have presented several positive results on this problem in addition to classical results which can be seen as special cases~\cite{Jerrum85}.
Namely, graph classes for which Token Swapping can be solved in polynomial-time are paths, cycles, complete graphs and complete bipartite graphs.
They showed that Token Swapping for general graphs belongs to NP.
The NP-hardness is recently shown in preliminary versions~\cite{KawaharaSY16a,KawaharaSY17} of this paper and by Miltzow et al.~\cite{MiltzowNORTU16} and Bonnet et al.~\cite{BonnetMR16} independently.
On the other hand, some polynomial-time approximation algorithms are known for different classes of graphs including the general case~\cite{HeathV03,YamanakaDIKKOSSUU14,MiltzowNORTU16}.
Our NP-hardness result is tight with respect to the degree bound, as the problem can be solved in polynomial-time if input graphs have vertex degree at most 2.
\begin{itemize}
\item Token swapping is NP-complete even when graphs are restricted to bipartite graphs where every vertex has degree at most 3 (Theorem~\ref{thm:TSPNPhard}).
\end{itemize}
Moreover, we present two polynomial-time solvable subcases of Token Swapping.
One is of lollipop graphs, which are combinations of a complete graph and a path.
The other is the class of graphs which are combinations of a star and a path.

A variant of Token Swapping is \emph{$c$-Colored Token Swapping}.
Tokens and vertices in this setting are colored by one of $c$ admissible colors.
We decide how many swaps are required to relocate the tokens so that each vertex has a token of the same color. 
Yamanaka et al.~\cite{YamanakaHKOSUU15} have shown that $3$-Colored Token Swapping is NP-complete while $2$-Colored Token Swapping is solvable in polynomial time.
This problem and a further generalization are also studied in~\cite{BonnetMR16}.
In this paper we consider the colored version of Routing via Matching and show that it is also NP-complete.
\begin{itemize}
	\item $2$-Coloring Routing via Matching is NP-complete even to decide whether an instance admits a $3$-step solution (Theorem~\ref{thm:2PCTSP}).
	\item $3$-Coloring Routing via Matching is NP-complete even to decide whether an instance admits a $2$-step solution (Theorem~\ref{thm:3PCTSP2}).
\end{itemize}
The former result contrasts the fact that the $2$-Colored Token Swapping is solvable in polynomial-time~\cite{YamanakaHKOSUU15}.
The latter contrasts that to decide whether $2$-step solution exists for Permutation Routing is in P (Theorem~\ref{thm:2PTSPPoly}).
In addition, we present another contrastive result.
\begin{itemize}
	\item It is decidable in polynomial-time whether a 2-step solution exists in $2$-Coloring Routing via Matching (Theorem~\ref{thm:22CRM}).
\end{itemize}

One may consider permutation routing on graphs as a special case of the \emph{Minimum Generator Sequence Problem}~\cite{EvenG81}.
The problem is to determine whether one can obtain a permutation $f$ on a finite set $X$ by multiplying at most $k$ permutations from a finite permutation set $\Pi$,
where all of $X$, $f$, $k$ and $\Pi$ are input.
The problem is known to be PSPACE-complete if $k$ is specified in binary notation~\cite{Jerrum85}, while it becomes NP-complete if $k$ is given in unary representation~\cite{EvenG81}.
In our settings, permutation sets $\Pi$ are restricted to the ones that have a graph representation.
However, this does not necessarily mean that the NP-hardness of Permutation Routing via Matching implies the hardness of the Minimum Generator Sequence Problem, since the description size of all the admissible parallel swaps on a graph is exponential in the graph size.

\tabResults{
Known and new results on the TSP and the variants are summarized in Table~\ref{tab:results}.
\begin{table}\small
\begin{tabular}{|c|c|c|}
\hline
Graph class	&	TSP	&	PTSP
\\ \hline
General (bipartite)	&	NP-hard [0]	&	NP-hard [0]
\\ \hline
Complete	&	P~\cite{Jerrum85}	&	P [0]
\\ \hline
Complete bipartite	&	P~\cite{YamanakaDIKKOSSUU14}	&	---
\\ \hline
Tree	&	\scalebox{0.9}[1]{PTime $2$-approx.}~\cite{YamanakaDIKKOSSUU14}	&	---
\\ \hline
Square-path	&	\scalebox{0.9}[1]{PTime $2$-approx.}~\cite{HeathV03}	&	---
\\ \hline
Lollipop	&	P [0]	&	---
\\ \hline
Star-path	&	P [0]	&	---
\\ \hline
Cycle	&	P~\cite{Jerrum85}	&	---
\\ \hline
Path	&	P~\cite{Jerrum85}	&	\scalebox{0.9}[1]{PTime $(+ 1)$-approx.} [0]
\\ \hline
\end{tabular}
\begin{tabular}{|c|c|c|}
\hline
$k$	&	$k$-CTSP	&	$k$-PCTSP
\\ \hline
$\ge 3$ 	&	NP-hard~\cite{YamanakaHKOSUU15}	&	NP-hard [0]
\\ \hline
$= 2$ 	&	P~\cite{YamanakaHKOSUU15}	&	NP-hard [0]
\\ \hline
\end{tabular}
\caption{\label{tab:results}Complexity of the TSP and variants on different graph classes}
\end{table}
}

\section{Time Complexity of Token Swapping}
We denote by $G=(V,E)$ an undirected graph whose vertex set is $V$ and edge set is $E$.
More precisely, elements of $E$ are subsets of $V$ consisting of exactly two distinct elements.
A \emph{configuration} $f$ (on $G$) is a permutation on $V$, i.e., bijection from $V$ to $V$.
By $[u]_f$ we denote the orbit $\{\, f^{i}(u) \mid i \in \mbb{N} \,\}$ of $u \in V$ under $f$.
We call each element of $V$ a \emph{token} when we emphasize the fact that it is in the range of $f$.
We say that a token $v$ \emph{is on a vertex $u$ in $f$} if $v=f(u)$.
A \emph{swap} on $G$ is a synonym for an edge of $G$, which behaves as a transposition.
For a configuration $f$ and a swap $e \in E$, the configuration obtained by applying $e$ to $f$, which we denote by $fe$, is defined by 
\[
fe(u) = \begin{cases}
f(v) & \text{if $e=\{u,v\}$,}
\\
f(u)	& \text{otherwise.}
\end{cases}
\]
For a sequence $\vec{e}=\lrangle{e_1,\dots,e_m}$ of swaps, the length $m$ is denoted by $|\vec{e}|$.
For $i \le m$, by $\vec{e}_{|\le i}$ we denote the prefix $\lrangle{e_1,\dots,e_i}$.
The configuration $f \vec{e}$ obtained by applying $\vec{e}$ to $f$ is $(\dots((f e_1)e_2)\dots)e_m$.
We say that the token $f(u)$ on $u$ \emph{is moved to} $v$ by $\vec{e}$ if $f\vec{e}(v)=f(u)$.
We count the total moves of each token $u \in V$ in the application as
\[
	\msf{move}(f,\vec{e},u) = |\{\, i \in \{1,\dots,m\} \mid (f\vec{e}_{|\le i-1})^{-1}(u) \neq (f\vec{e}_{|\le i})^{-1}(u) \,\}|\,.
\]
Clearly $\msf{move}(f,\vec{e},u) \ge \msf{dist}(f^{-1}(u),(f\vec{e})^{-1}(u))$, where $\msf{dist}(u_1,u_2)$ denotes the length of a shortest path between $u_1$ and $u_2$,
and $\sum_{u \in V}\msf{move}(f,\vec{e},u) = 2|\vec{e}|$. 

We denote the set of \emph{solutions} for a configuration $f$ by
\[
	\SOL(G,f) = \{\, \vec{e} \mid \vec{e} \text{ is a swap sequence on $G$ such that } f \vec{e} \text{ is the identity}\,\}\,.
\]
A solution $\vec{e}_0 \in \SOL(G,f)$ is said to be \emph{optimal} if $|\vec{e}_0| =  \min \{\,|\vec{e}| \mid \vec{e} \in \SOL(G,f)\,\}$.
The length of an optimal solution is denoted by $\OPT(G,f)$.

\problemdef{Token Swapping}%
{A connected graph $G$, a configuration $f$ on $G$ and a natural number $k$.}%
{$\OPT(G,f) \le k$?}

\subsection{Token Swapping Is NP-complete}\label{sec:TSP-NPhard}
This subsection proves the NP-hardness of Token Swapping by a reduction from the 3DM, which is known to be NP-complete~\cite{Karp72}.

\problemdef{Three dimensional matching problem, 3DM}%
{Three disjoint sets $A_1,A_2,A_3$ such that $|A_1|=|A_2|=|A_3|$ and a set $T \subseteq A_1 \times A_2 \times A_3$.}%
{Is there $M \subseteq T$ such that $|M|=|A_1|$ and every element of $A_1 \cup A_2 \cup A_3$ occurs just once in $M$?} 

An instance of the 3DM is denoted by $(A,T)$ where $A = A_1 \cup A_2 \cup A_3$ assuming that the partition is understood.
Let $A_k=\{a_{k,1},\dots,a_{k,n}\}$ for $k \in \{1,2,3\}$ and $T = \{t_1,\dots,t_m\}$.
For notational convenience we write $a \in t$ if $a \in A$ occurs in $t \in T$ by identifying $t$ with the set of the elements of $t$.
We construct an instance $(G_T,f)$ of Token Swapping as follows.
The vertex set of $G_T$ is $V_A \cup V_T$ with
\begin{align*}
	V_A &= \{\,u_{k,i},u_{k,i}'\mid k \in \{1,2,3\} \text{ and } i \in \{1,\dots,n\} \}\,,
\\	V_T &= \{\,v_{j,k}, v_{j,k}' \mid j \in \{1,\dots,m\} \text{ and } k \in \{1,2,3\} \}\,.
\end{align*}
The edge set $E_T$ is given by
\begin{align*}
	E_T ={} & \{\, \{u_{k,i},v_{j,k}'\},\{u_{k,i}',v_{j,k}\} \mid \text{$a_{k,i} \in A_k$ occurs in $t_j \in T$} \,\}
	\\ &{} \cup \{\, \{v_{j,k},v_{j,l}'\} \subseteq V_T \mid j \in \{1,\dots,m\} \text{ and } k \neq l \,\}\,.
\end{align*}
We call the subgraph induced by $\{v_{j,1},v_{j,2}',v_{j,3},v_{j,1}',v_{j,2},v_{j,3}'\}$ the \emph{$t_j$-cycle}.
The initial configuration $f$ is defined by
\begin{align*}
	f(u_{k,i}) &= u_{k,i}' \text{ and } f(u_{k,i}') = u_{k,i} \text{ for all $a_{k,i} \in A_k$ and $k \in \{1,2,3\}$}\,,
\\	
	f(v_{j,k}) &= v_{j,k} \text{ and } f(v_{j,k}') = v_{j,k}' \text{ for all $t_j \in T$ and $k \in \{1,2,3\}$\,.}
\end{align*}
\begin{figure}[h]
\begin{center}
\tikzstyle{vertex}=[circle,draw,inner sep=0pt,minimum size=6mm]
\tikzstyle{token}=[rectangle,draw,fill=white,inner sep=0pt,minimum size=5.5mm]
\begin{tikzpicture}[scale=0.85]\small
\node (t1-) at (2.2,6.5) [vertex] {$v_{1,1}'$};
\node (t2) at (3.1,7) [vertex] {$v_{1,2}$};
\node (t3-) at (4,6.5) [vertex] {$v_{1,3}'$};
\node (t1) at (4,5.5) [vertex] {$v_{1,1}$};
\node (t2-) at (3.1,5) [vertex] {$v_{1,2}'$};
\node (t3) at (2.2,5.5) [vertex] {$v_{1,3}$};
\node (u1-) at (6.1,6.5) [vertex] {$v_{2,1}'$};
\node (u2) at (7,7) [vertex] {$v_{2,2}$};
\node (u3-) at (7.9,6.5) [vertex] {$v_{2,3}'$};
\node (u1) at (7.9,5.5) [vertex] {$v_{2,1}$};
\node (u2-) at (7,5) [vertex] {$v_{2,2'}$};
\node (u3) at (6.1,5.5) [vertex] {$v_{2,3}$};
\node (v1-) at (10,6.5) [vertex] {$v_{3,1}'$};
\node (v2) at (10.9,7) [vertex] {$v_{3,2}$};
\node (v3-) at (11.8,6.5) [vertex] {$v_{3,3}'$};
\node (v1) at (11.8,5.5) [vertex] {$v_{3,1}$};
\node (v2-) at (10.9,5) [vertex] {$v_{3,2}'$};
\node (v3) at (10,5.5) [vertex] {$v_{3,3}$};
\draw [-] (t1-) -- (t2) -- (t3-) -- (t1) -- (t2-) -- (t3) -- (t1-);
\draw [-] (u1-) -- (u2) -- (u3-) -- (u1) -- (u2-) -- (u3) -- (u1-);
\draw [-] (v1-) -- (v2) -- (v3-) -- (v1) -- (v2-) -- (v3) -- (v1-);
\node (a1) at (2.5,8.4) [vertex] {$u_{1,1}$};
\node (a2) at (4.3,8.4) [vertex] {$u_{1,2}$};
\node (b1-) at (6.1,8.4) [vertex] {$u_{2,1}'$};
\node (b2-) at (7.9,8.4) [vertex] {$u_{2,2}'$};
\node (c1) at (9.7,8.4) [vertex] {$u_{3,1}$};
\node (c2) at (11.5,8.4) [vertex] {$u_{3,2}$};
\node (a1-) at (9.7,3.6) [vertex] {$u_{1,1}'$};
\node (a2-) at (11.5,3.6) [vertex] {$u_{1,2}'$};
\node (b1) at (6.1,3.6) [vertex] {$u_{2,1}$};
\node (b2) at (7.9,3.6) [vertex] {$u_{2,2}$};
\node (c1-) at (2.5,3.6) [vertex] {$u_{3,1}'$};
\node (c2-) at (4.3,3.6) [vertex] {$u_{3,2}'$};
%
\draw [-] (t1) --  (a1-);
\draw [-] (t1-) -- (a1);
\draw [-] (t2) -- (b1-);
\draw [-] (t2-) -- (b1);
\draw [-] (t3) -- (c1-);
\draw [-] (t3-) -- (c1);
%
\draw [-] (u1) -- (a1-);
\draw [-] (u1-) -- (a1);
\draw [-] (u2) -- (b1-);
\draw [-] (u2-) -- (b1);
\draw [-] (u3) -- (c2-);
\draw [-] (u3-) -- (c2);
\draw [-] (v1) -- (a2-);
\draw [-] (v1-) -- (a2);
\draw [-] (v2) -- (b2-);
\draw [-] (v2-) -- (b2);
\draw [-] (v3) -- (c2-);
\draw [-] (v3-) -- (c2);
\node at (2.5,9.0) [token] {$u_{1,1}'$};
\node at (4.3,9.0) [token] {$u_{1,2}'$};
\node at (6.1,9.0) [token] {$u_{2,1}$};
\node at (7.9,9.0) [token] {$u_{2,2}$};
\node at (9.7,9.0) [token] {$u_{3,1}'$};
\node at (11.5,9.0) [token] {$u_{3,2}'$};
\node at (2.5,3.0) [token] {$u_{3,1}$};
\node at (4.3,3.0) [token] {$u_{3,2}$};
\node at (6.1,3.0) [token] {$u_{2,1}'$};
\node at (7.9,3.0) [token] {$u_{2,2}'$};
\node at (9.7,3.0) [token] {$u_{1,1}$};
\node at (11.5,3.0) [token] {$u_{1,2}$};
\end{tikzpicture}
\end{center}
\caption{\label{fig:3DM-TSP}
The graph and initial configuration of Token Swapping reduced from the 3DM instance in Example~\ref{ex:3dm-TSP}.
Vertices and tokens are denoted by circles and squares, respectively. The tokens which are already on the goal vertices in the initial configuration are omitted.}
\end{figure}
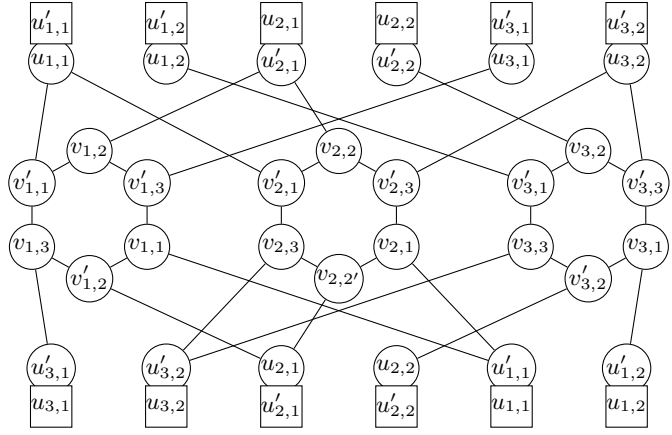

In the initial configuration $f$, all and only the tokens in $V_A$ are misplaced.
Each token $u_{k,i} \in V_A$ on the vertex $u_{k,i}'$ must be moved to $u_{k,i}$ via (a part of) $t_j$-cycle for some $t_j \in T$ in which $a_{k,i}$ occurs.
To design a short solution for $(G_T,f)$, it is desirable to have swaps at which both of the swapped tokens get closer to the destination.
If $(A,T)$ admits a solution, then one can find an optimal solution for $(G_T,f)$ of length $21n$, where $9n$ of the swaps satisfy this property as we will see in Lemma~\ref{lem:TSP-NPhard_upperbound}.
On the other hand, such an ``efficient'' solution is possible only when $(A,T)$ admits a solution as shown in Lemma~\ref{lem:TSP-NPhard_lowerbound}.
\begin{example}\label{ex:3dm-TSP}
Let $A = A_1 \cup A_2 \cup A_3$ and $T=\{t_1,t_2,t_3\}$ 
where $A_k=\{a_{k,1},a_{k,2}\}$ for $k \in \{1,2,3\}$, $t_1=\{a_{1,1},a_{2,1},a_{3,1}\}$, $t_2=\{a_{1,1},a_{2,1},a_{3,2}\}$ and $t_3= \{a_{1,2},a_{2,2},a_{3,2}\}$.
Figure~\ref{fig:3DM-TSP} shows the graph and initial configuration reduced from the 3DM instance $(A,T)$.
This instance $(A,T)$ has a solution $M=\{t_1,t_3\}$.
The proof of Lemma~\ref{lem:TSP-NPhard_upperbound} will give how to find an optimal solution for the reduced TSP instance corresponding to $M$.
A part of the solution is illustrated in Figure~\ref{fig:3DM-TSP2}.
\end{example}
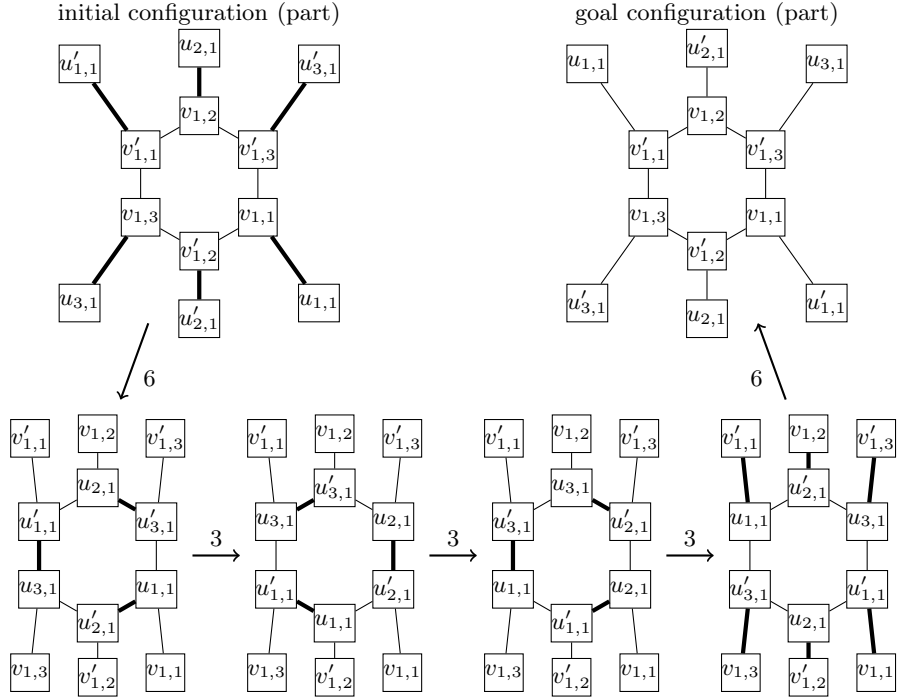
\begin{figure}
\begin{center}
\tikzstyle{token}=[rectangle,draw,fill=white,inner sep=0pt,minimum size=5mm]
\newcommand{\configA}[8]{
\node (x) at (90:2.5) [] {#8 (part)};
\node (t1-) at (150:1) [token] {$v_{1,1}'$};
\node (t2) at (90:1) [token] {$v_{1,2}$};
\node (t3-) at (30:1) [token] {$v_{1,3}'$};
\node (t1) at (-30:1) [token] {$v_{1,1}$};
\node (t2-) at (-90:1) [token] {$v_{1,2}'$};
\node (t3) at (-150:1) [token] {$v_{1,3}$};
\draw [-] (t1-) -- (t2);
\draw [-] (t2) -- (t3-);
\draw [-] (t3-) -- (t1);
\draw [-] (t1) -- (t2-);
\draw [-] (t2-) -- (t3);
\draw [-] (t3) -- (t1-);
\node (a1) at (135:2.5) [token] {#1};
\node (b1-) at (90:2) [token] {#2};
\node (c1) at (45:2.5) [token] {#3};
\node (a1-) at (-45:2.5) [token] {#4};
\node (b1) at (-90:2) [token] {#5};
\node (c1-) at (-135:2.5) [token] {#6};
%
\draw [#7] (t1) --  (a1-);
\draw [#7] (t1-) -- (a1);
\draw [#7] (t2) -- (b1-);
\draw [#7] (t2-) -- (b1);
\draw [#7] (t3) -- (c1-);
\draw [#7] (t3-) -- (c1);
}
\newcommand{\configB}[9]{
\node (t1-) at (150:1) [token] {#1};
\node (t2) at (90:1) [token] {#2};
\node (t3-) at (30:1) [token] {#3};
\node (t1) at (-30:1) [token] {#4};
\node (t2-) at (-90:1) [token] {#5};
\node (t3) at (-150:1) [token] {#6};
\draw [#8] (t1-) -- (t2);
\draw [#9] (t2) -- (t3-);
\draw [#8] (t3-) -- (t1);
\draw [#9] (t1) -- (t2-);
\draw [#8] (t2-) -- (t3);
\draw [#9] (t3) -- (t1-);
\node (a1) at (120:2.0) [token] {$v_{1,1}'$};
\node (b1-) at (90:1.8) [token] {$v_{1,2}$};
\node (c1) at (60:2.0) [token] {$v_{1,3}'$};
\node (a1-) at (-60:2.0) [token] {$v_{1,1}$};
\node (b1) at (-90:1.8) [token] {$v_{1,2}'$};
\node (c1-) at (-120:2.0) [token] {$v_{1,3}$};
%
\draw [#7] (t1) --  (a1-);
\draw [#7] (t1-) -- (a1);
\draw [#7] (t2) -- (b1-);
\draw [#7] (t2-) -- (b1);
\draw [#7] (t3) -- (c1-);
\draw [#7] (t3-) -- (c1);
}
\begin{tikzpicture}[scale=0.90]\small
\configA{$u_{1,1}'$}{$u_{2,1}$}{$u_{3,1}'$}{$u_{1,1}$}{$u_{2,1}'$}{$u_{3,1}$}{ultra thick}{initial configuration}
\draw [thick, ->] (-110:2.2) -- node[auto] {6} (-110:3.4);
\pgftransformxshift{-1.5cm}
\pgftransformyshift{-5.5cm}
\configB{$u_{1,1}'$}{$u_{2,1}$}{$u_{3,1}'$}{$u_{1,1}$}{$u_{2,1}'$}{$u_{3,1}$}{-}{-}{ultra thick}
\draw [thick, ->] (0:1.4) -- node[auto] {3} (0:2.1);
\pgftransformxshift{3.5cm}
\configB{$u_{3,1}$}{$u_{3,1}'$}{$u_{2,1}$}{$u_{2,1}'$}{$u_{1,1}$}{$u_{1,1}'$}{-}{ultra thick}{-}
\draw [thick, ->] (0:1.4) -- node[auto] {3} (0:2.1);
\pgftransformxshift{3.5cm}
\configB{$u_{3,1}'$}{$u_{3,1}$}{$u_{2,1}'$}{$u_{2,1}$}{$u_{1,1}'$}{$u_{1,1}$}{-}{-}{ultra thick}
\draw [thick, ->] (0:1.4) -- node[auto] {3} (0:2.1);
\pgftransformxshift{3.5cm}
\configB{$u_{1,1}$}{$u_{2,1}'$}{$u_{3,1}$}{$u_{1,1}'$}{$u_{2,1}$}{$u_{3,1}'$}{ultra thick}{-}{-}
\pgftransformyshift{5.5cm}
\pgftransformxshift{-1.5cm}
\draw [thick, ->] (-70:3.4) -- node[auto] {6} (-70:2.2);
\configA{$u_{1,1}$}{$u_{2,1}'$}{$u_{3,1}$}{$u_{1,1}'$}{$u_{2,1}$}{$u_{3,1}'$}{-}{goal configuration}
\end{tikzpicture}
\end{center}
\caption{\label{fig:3DM-TSP2}
The 3DM instance $(A,T)$ of Example~\ref{ex:3dm-TSP} has a solution $M=\{t_1,t_3\}$.
The optimal solution given in the proof of Lemma~\ref{lem:TSP-NPhard_upperbound} that exchanges $u_{k,1}$ and $u_{k,1}'$ for all $k \in \{1,2,3\}$ via the $t_1$-cycle is illustrated here, where we suppress vertex names.
By swapping the tokens on the bold edges in each configuration, we obtain the succeeding one pointed by an arrow.
The number by each arrow shows the number of swaps. The swap sequence consists of 21 swaps in total.
By doing the same on $t_3$-cycle with respect to $u_{1,2},u_{2,2},u_{3,2},u_{1,2}',u_{2,2}',u_{3,2}'$, we obtain the goal configuration.}
\end{figure}

\begin{lemma}\label{lem:TSP-NPhard_upperbound}
	If\/ $(A,T)$ has a solution then $\OPT(G_T,f) \le 21n$ with $n = |A_1|$.
\end{lemma}
\begin{proof}
	We show in the next paragraph that for each $t_j \in T$, there is a sequence $\sigma_j$ of $21$ swaps such that
	$g\sigma_j$ is identical to $g$ except $(g\sigma_j) (u_{k,i}) = g(u_{k,i}')$ and $(g\sigma_j)(u_{k,i}')=g(u_{k,i})$ if $a_{k,i}$ occurs in $t_{j}$ for any configuration $g$.
	If $M \subseteq T$ is a solution, by collecting $\sigma_j$ for all $t_j \in M$, we obtain a swap sequence $\sigma_M$ of length $21n$
	such that $f\sigma_M$ is the identity.

	Let $t_j = \{a_{1,i_1},a_{2,i_2},a_{3,i_3}\}$.
	We first move each of the tokens $u_{k,i_k}$ on the vertex $u_{k,i_k}'$ to the vertex $v_{j,k}$
	and the tokens $u_{k,i_k}'$ on $u_{k,i_k}$ to $v_{j,k}'$.
	We then move the tokens $u_{k,i_k}$ on $v_{j,k}$ to the opposite vertex $v_{j,k}'$ of the $t_j$-cycle for each $k \in \{1,2,3\}$ 
	while moving $u_{k,i_k}'$ on $v_{j,k}'$ to $v_{j,k}$ in the opposite direction simultaneously.
	At last we make swaps on the same $6$ edges we used in the first phase. 
	The above procedure consists of $21$ swaps and gives the desired configuration.
\end{proof}
\begin{lemma}\label{lem:TSP-NPhard_lowerbound}
	If\/ $\OPT(G_T,f) \le 21n$ with $n = |A_1|$ then $(A,T)$ has a solution.
\end{lemma}
\begin{proof}
	We first show that $21n$ is a lower bound on $\OPT(G_T,f)$.
	Suppose that $f \sigma$ is the identity.
	For each token $u_{k,i} \in V_A$, we have 
	\[
	\msf{move}(f,\sigma,u_{k,i}) \ge \msf{dist}(u_{k,i},f^{-1}(u_{k,i})) = \msf{dist}(u_{k,i},u_{k,i}') = 5\,.
	\]
	The adjacent vertices of the vertex $u_{k,i}'$ are $v_{j,k}$ such that $a_{k,i} \in t_j$.
	Among those, let $\tau(u_{k,i}) \in V_T$ be the vertex to which $u_{k,i}$ goes for its first step, i.e., the first occurrence of $u_{k,i}'$ in $\sigma$ is as $\{u_{k,i}',\tau(u_{k,i})\}$.
	This means that $\msf{move}(f,\sigma,\tau(u_{k,i})) \ge 2$, since the token $\tau(u_{k,i})$ must once leave from and later come back to the vertex $\tau(u_{k,i})$.
	The symmetric discussion holds for all tokens $u_{k,i}'$.
	Therefore, noting that $\tau$ is an injection, we obtain
	\begin{align*}
		|\sigma| &= \frac{1}{2} \sum_{x \in V_A \cup V_T}\msf{move}(f,\sigma,x) 
		\ge \frac{1}{2} \sum_{x \in V_A} \big(\msf{move}(f,\sigma,x)+\msf{move}(f,\sigma,\tau(x))\big)
		  \ge 21n\,.
	\end{align*}
	This has shown that if $f \sigma $ is the identity and $|\sigma| \le 21n$, then 
	\begin{enumerate}
		\item[(1)] $\msf{move}(f,\sigma,x)=5$ for all $x \in V_A$,
		\item[(2)] $\msf{move}(f,\sigma,y) \neq 0$ for $y \in V_T$ if and only if $y = \tau(x)$ for some $x \in V_A$.
	\end{enumerate}	
	Let \(	M_\sigma = \{\, y \in V_T \mid \msf{move}(f,\sigma,y) \neq 0 \,\} = \{\, \tau(x) \in V_T \mid x \in V_A\,\} \).
	We are now going to prove that if $v_{j,1} \in M_\sigma$ then 
	$\{v_{j,2},v_{j,3},v_{j,1}',v_{j,2}',v_{j,3}'\} \subseteq M_\sigma$, which implies that 
\(
		\widetilde{M}_\sigma = \{\, t_j \in T \mid v_{j,1} \in M_\sigma \,\} 
\)
 is a solution for $(A,T)$.
	
	Suppose $v_{j,1} \in M_\sigma$ and let $t_j \cap A_1 = \{a_{1,i}\}$.
	This means that $\tau(u_{1,i})=v_{j,1}$ and $u_{1,i}$ goes from $u_{1,i}'$ to $u_{1,i}$ through $(u_{1,i}',v_{j,1},v_{j,2}',v_{j,3},v_{j,1}',u_{1,i})$ or $(u_{1,i}',v_{j,1},v_{j,3}',\linebreak[0]v_{j,2},v_{j,1}',u_{1,i})$ by ({2}) and ({1}).
	In either case, $v_{j,1}' \in M_\sigma$.
	Suppose that $u_{1,i}$ takes the former $(u_{1,i}',v_{j,1},v_{j,2}',v_{j,3},v_{j,1}',u_{1,i})$.
	Then $v_{j,2}',v_{j,3} \in M_\sigma$.
	Just like $v_{j,1} \in M_\sigma $ implies $v_{j,1}' \in M_\sigma$, we now see $v_{j,2},v_{j,3}' \in M_\sigma$.
\end{proof}
It is known that the 3DM is still NP-complete if each $a \in A$ occurs at most three times in $T$~\cite{GareyJ79}.
Assuming that $T$ satisfies this constraint, it is easy to see that $G_T$ is a bipartite graph with maximum vertex degree 3.
\begin{theorem}\label{thm:TSPNPhard}
Token swapping is NP-complete even on bipartite graphs with maximum vertex degree 3.
\end{theorem}
The NP-hardness of Token Swapping was independently proven by Miltzow et al.~\cite{MiltzowNORTU16} and by Bonnet et al.~\cite{BonnetMR16}.
The graphs obtained by the reduction of Miltzow et al.\ have a degree bound but it is not as small as our constraint.
Our bound $3$ is tight, as Token Swapping on graphs with degree at most 2, i.e., paths and cycles, is solvable in polynomial-time.
Bonnet et al.~\cite{BonnetMR16} have given no degree constraint but their graphs have tree-width 2 and diameter 6.
Therefore, their and our results are incomparable.

\subsection{PTIME Subcases of Token Swapping}
In this subsection, we present two graph classes on which Token Swapping can be solved in polynomial time.
One is that of \emph{lollipop graphs}, which are obtained by connecting a path and a complete graph with a bridge.
That is, a {lollipop graph} is $L_{m,n} = (V,E)$ where
\(	V = \{\, -m,\dots,-1,0,1,\dots,n\,\} \) and
\begin{align*}
	E &= \{\, \{i,j\} \subseteq V \mid i < j \le 0 \text{ or } j = i+1 > 0 \,\}\,.
\end{align*}
The other class consists of graphs obtained by connecting a path and a star.
A \emph{star-path graph} is $Q_{m,n} = (V,E)$ such that
\( V = \{\, -m,\dots,-1,0,1,\dots,n\,\}\) and 
\begin{align*}
	E &= \{\, \{i,0\} \subseteq V \mid i < 0 \,\} \cup \{\, \{i,i+1\} \subseteq V \mid i \ge 0 \,\}\,.
\end{align*}
Algorithms~\ref{alg:lollipop} and~\ref{alg:gerbera} give optimal solutions for Token Swapping on lollipop and star-path graphs in polynomial time, respectively.
Proofs of the correctness are found in Appendices~\ref{sec:lollipop} and~\ref{sec:gerbera}.
\begin{algorithm}\caption{Algorithm for Token Swapping on Lollipop Graphs\label{alg:lollipop}}
\begin{algorithmic}
  \STATE \textbf{Input}: A lollipop graph $L_{m,n}$ and a configuration $f$ on $L_{m,n}$
  \FOR{$k = n,\dots,1,0,-1,\dots,-m$}
  	\STATE Move the token $k$ to the vertex $k$ directly;
  \ENDFOR
   \end{algorithmic}
\end{algorithm}
\begin{algorithm}\caption{Algorithm for Token Swapping on Star-Path Graphs\label{alg:gerbera}}
\begin{algorithmic}
  \STATE \textbf{Input}: A star-path graph $Q_{m,n}$ and a configuration $f$ on $Q_{m,n}$
  \FOR{$k = n,\dots,1,0,-1,\dots,-m$}
  	\WHILE{the token on the vertex $0$ has an index less than $0$}
		\STATE Move the token on the vertex $0$ to its goal vertex;
	\ENDWHILE
  	\STATE Move the token $k$ to the vertex $k$;
  \ENDFOR
   \end{algorithmic}
\end{algorithm}

\section{Permutation Routing via Matching}
\emph{Permutation routing via matching} can be seen as the parallel version of Token Swapping.
Definitions and notation for Token Swapping are straightforwardly generalized as follows. 
A \emph{parallel swap} $S$ on $G$ is a synonym for an involution which is a subset of $E$, or for a matching of $G$, i.e., $S \subseteq E$ such that $\{u,v_1\},\{u,v_2\} \in S$ implies $v_1=v_2$.
For a configuration $f$ and a parallel swap $S \subseteq E$, the configuration obtained by applying $S$ to $f$ is defined by 
$fS(u) = f(v)$ if $\{u,v\} \in S$ and $fS(u)=f(u)$ if $u \notin \bigcup S$.
Let
\begin{align*}
	\PSOL(G,f) &= \{\, \vec{S} \mid \vec{S} \text{ is a parallel swap sequence s.t.\ } f \vec{S} \text{ is the identity}\,\}
\\	\POPT(G,f) &= \min\{\, |\vec{S}| \mid \vec{S} \in \PSOL(G,f)\,\}\,.
\end{align*}
\problemdef{Permutation Routing via Matching}%
{A connected graph $G$, a configuration $f$ on $G$ and a natural number $k$.}%
{$\POPT(G,f) \le k$?}

It is trivial that $\POPT(G,f) \le \OPT(G,f) \le \POPT(G,f)|V|/2$, since any parallel swap $S$ consists of at most $|V|/2$ (single) swaps.
Since $ \OPT(G,f) \le |V|(|V|-1)/2$ holds~\cite{YamanakaDIKKOSSUU14}, Permutation Routing via Matching belongs to NP.


We use the following easy result in many places in this section.
\begin{lemma}\label{lem:trivial}
Let $P_n$ denote the path graph with $n$ vertices, i.e., $P_n=(\{1,\dots,n\},\{\, \{i,i+1\} \mid 1 \le i< n\,\})$,
and a configuration $f$ be the identity configuration on $P_n$  except $f(1)=n$ and $f(n)=1$.
Then \[
\POPT(P_n,f) = \begin{cases}
n-1	& \text{ if $n$ is even,}
\\
n	& \text{ if $n$ is odd.}
\end{cases}
\]
\end{lemma}
\begin{proof}
	We have $\POPT(P_n,f) \ge \msf{dist}(1,f^{-1}(1)) = n-1$.
	If $n$ is even, it is easy to check that $\lrangle{S_1,\dots,S_{n-1}} \in \PSOL(P_n,f)$ where $S_i = \{\, \{i,i+1\}, \{n-i,n-i+1\}\,\}$ for all $i \in \{1,\dots,n-1\}$.
	Suppose that $n$ is odd.
	Then $\lrangle{S_1,\dots,S_n} \in \PSOL(P_n,f)$ where $S_1=S_n=\{1,2\}$ and $S_i=\{\, \{i,i+1\}, \{n-i+1,n-i+2\}\,\}$ for all $i \in \{2,\dots,n-1\}$.
	This shows $\POPT(P_n,f) \le n$.
	To derive a contradiction, suppose there is $\lrangle{S_1,\dots,S_{n-1}} \in \PSOL(P_n,f)$.
	To move the token $1$ on the vertex $n$ to the goal $1$ within $n-1$ steps, we must have $\{n-i,n-i+1\} \in S_i$ for $i=1,\dots,n-1$.
	To move the token $n$ on the vertex $1$ to the goal $n$ within $n-1$ steps, we must have $\{i,i+1\} \in S_i$ for $i=1,\dots,n-1$.
	However, this means $\{ \ceil{n/2}-1,\ceil{n/2}\},\{ \ceil{n/2},\ceil{n/2}+1\} \in S_{\ceil{n/2}}$, which is impossible.
\end{proof}

\subsection{Routing Permutations via Matching Is NP-complete}\label{sec:PTSPNPhard}
We show that routing permutations via matching is NP-hard by a reduction from a restricted kind of the satisfiability problem, which we call \emph{PPN-Separable 3SAT} (\emph{\SAT} for short).
For a set $X$ of \emph{(Boolean) variables}, $\neg X$ denotes the set of their negative literals.
A \emph{3-clause} is a subset of $X \cup \neg X$ whose cardinality is at most 3.
An instance of \SAT{} is a finite collection ${F}$ of $3$-clauses,
which can be partitioned into three subsets ${F}_1, {F}_2,{F}_3 \subseteq {F}$
such that for each variable $x \in X$, the positive literal $x$ occurs just once in each of ${F}_1,{F}_2$ and never in $F_3$, and the negative literal $\neg x$ occurs just once in ${F}_3$ and never in $F_1$ nor $F_2$.
Note that one can find a partition $\{F_1 , F_2, F_3\}$ of a \SAT{} instance $F$ in linear time.

\begin{theorem}\label{thm:SAT}
\SAT{} is NP-complete.
\end{theorem}
\begin{proof}
	See Appendix~\ref{sec:SAT}.
\end{proof}
We give a reduction from \SAT{} to Permutation Routing via Matching.
For a given instance ${F} = \{C_1,\dots,C_n\}$ over a variable set $X = \{x_1,\dots,x_m\}$ of \SAT{}, we define a graph $G_{F} = (V_{F},E_{F})$ in the following manner.
Let ${F}$ be partitioned into ${F}_1,{F}_2,{F}_3$ where each of ${F}_1$ and ${F}_2$ has just one occurrence of each variable as a positive literal and ${F}_3$ has just one occurrence of each negative literal.
Define
\begin{align*}
	V_F = {}& \{\, u_i,u_i',u_{i,1},u_{i,2},u_{i,3},u_{i,4} \mid 1 \le i \le m\,\}
	\cup \{\, v_j,v_j' \mid 1 \le j \le n \,\} 
\,.\end{align*}
The edge set $E_{F}$ is the least set that makes $G_{F}$ contain the following paths of length $2$ and $3$:
\begin{gather*}
	(u_i,u_{i,1},u_{i,2},u_{i}')	
	\text{ and }
	(u_i,u_{i,3},u_{i,4},u_{i}')	 \text{ for each $i \in \{1,\dots,m\}$,}
\\
	(v_j,u_{i,k},v_{j}')	 \text{ if $x_i \in C_j \in {F}_k$ or $\neg x_i \in C_j \in {F}_k$}
\,.\end{gather*}
The fact that $G_{F}$ is bipartite can be seen by partitioning $V_{F}$ into 
\[
\{\, u_i,u_{i,2},u_{i,4} \mid 1 \le i \le m \,\} \cup \{\, v_j,v_j' \mid C_j \in F_1 \cup F_3 \,\}
\]
 and the rest.
Vertices $v_j$ and $v_j'$ have degree at most 3 for $j \in \{1,\dots,n\}$, while $u_{i,k}$ has degree 4 for $i \in \{1,\dots,m\}$ and $k \in \{1,2,3\}$.
The initial configuration $f$ is defined to be the identity except
\begin{gather*}
	f(u_i) = u_i',\
	f(u_i') = u_i,\
	f(v_j) = v_j',\
	f(v_j') = v_j,
\end{gather*}
for each $i \in \{1,\dots,m\}$ and $j \in \{1,\dots, n\}$.
\begin{example}\label{ex:SAT-PTSP}
For $X = \{x_1,x_2,x_3\}$, let ${F}$ consist of  $C_1 = \{x_1,x_2\}$, $C_2 = \{x_3\}$, $C_3 = \{x_1\}$, $C_4 = \{x_2,x_3\}$ and $C_5=\{\neg x_1,\neg x_2,\neg x_3\}$.
Then ${F}$ is partitioned into ${F}_1 = \{C_1,C_2\}$, ${F}_2=\{C_3,C_4\}$ and ${F}_3=\{C_5\}$, where each variable occurs just once in each ${F}_k$ with $k \in \{1,2,3\}$.
Moreover, ${F}_1$ and ${F}_2$ have only positive literals and ${F}_3$ has only negative literals.
Therefore, $F$ is a \SAT{} instance.
Figure~\ref{fig:SAT-PTSP} shows the reduction from $F$.
The formula ${F}$ is satisfied by assigning $1$ to $x_1,x_3$ and $0$ to $x_2$.
Corresponding to this assignment, by moving misplaced tokens along the bold edges in Figure~\ref{fig:SAT-PTSP}, the goal configuration is realized in 3 steps.
\end{example}

\begin{figure}
\begin{center}
\tikzstyle{vertex}=[circle,draw,inner sep=0pt,minimum size=6mm]
\tikzstyle{token}=[rectangle,draw,fill=white,inner sep=0pt,minimum size=4.5mm]
\begin{tikzpicture}[scale=0.9]\small
	\node (u1) at (0,0) [vertex] {$u_1$};
	\node (u11) at (0.7,0.9) [vertex] {$u_{1,1}$};
	\node (u12) at (1.7,0.9) [vertex] {$u_{1,2}$};
	\node (u13) at (0.7,-0.9) [vertex] {$u_{1,3}$};
	\node (u14) at (1.7,-0.9) [vertex] {$u_{1,4}$};
	\node (u1-) at (2.4,0) [vertex] {$u_1'$};
	\pgftransformxshift{5cm}
	\node (u2) at (0,0) [vertex] {$u_2$};
	\node (u21) at (0.7,0.9) [vertex] {$u_{2,1}$};
	\node (u22) at (1.7,0.9) [vertex] {$u_{2,2}$};
	\node (u23) at (0.7,-0.9) [vertex] {$u_{2,3}$};
	\node (u24) at (1.7,-0.9) [vertex] {$u_{2,4}$};
	\node (u2-) at (2.4,0) [vertex] {$u_2'$};
	\pgftransformxshift{5cm}
	\node (u3) at (0,0) [vertex] {$u_3$};
	\node (u31) at (0.7,0.9) [vertex] {$u_{3,1}$};
	\node (u32) at (1.7,0.9) [vertex] {$u_{3,2}$};
	\node (u33) at (0.7,-0.9) [vertex] {$u_{3,3}$};
	\node (u34) at (1.7,-0.9) [vertex] {$u_{3,4}$};
	\node (u3-) at (2.4,0) [vertex] {$u_3'$};
	\pgftransformxshift{-10cm}
	\draw (u1) -- (u11) -- (u12) -- (u1-);
	\draw[very thick]  (u1-) -- (u14) -- (u13) -- (u1);
	\draw[very thick]  (u2) -- (u21) -- (u22) -- (u2-);
	\draw (u2-) -- (u24) -- (u23) -- (u2);
	\draw (u3) -- (u31) -- (u32) -- (u3-);
	\draw[very thick] (u3-) -- (u34) -- (u33) -- (u3);
	\node (v1) at (3.2,4) [vertex] {$v_1$};
	\node (v1') at (3.2,3.2) [vertex] {$v_1'$};
	\draw[very thick]  (v1) .. controls +(180:2) and +(90:2) .. (u11) .. controls +(75:1) and +(180:1) .. (v1');
	\draw (v1) .. controls +(0:2) and +(90:2) .. (u21) .. controls +(105:1) and +(0:1) ..  (v1');
	\node (v2) at (9.7,2) [vertex] {$v_2$};
	\node (v2') at (8.7,1.1) [vertex] {$v_2'$};
	\draw[very thick] (v2) -- (u31) -- (v2');
	\node (v3) at (2.7,2) [vertex] {$v_3$};
	\node (v3') at (3.7,1.1) [vertex] {$v_3'$};
	\draw[very thick] (v3) -- (u12) -- (v3');
	\pgftransformxshift{6cm}
	\node (v4) at (3.2,4) [vertex] {$v_4$};
	\node (v4') at (3.2,3.2) [vertex] {$v_4'$};
	\draw (v4) .. controls +(180:2) and +(90:2) ..  (u22) .. controls +(75:1) and +(180:1) .. (v4');
	\draw[very thick]  (v4) .. controls +(0:2) and +(90:2) .. (u32) .. controls +(105:1) and +(0:1) ..  (v4');
	\pgftransformxshift{-6cm}
	\node (v5) at (7,-3) [vertex] {$v_5$};
	\node (v5') at (4.7,-1.8) [vertex] {$v_5'$};
	\draw (v5) .. controls +(180:2) and +(-60:2) .. (u13) .. controls +(-45:1) and +(180:2) .. (v5');
	\draw[very thick] (v5) -- (u23) -- (v5');
	\draw (v5) .. controls +(15:1) and +(240:2) ..  (u33) .. controls +(225:1) and +(0:2) ..  (v5');
	\node at (-0.5,0.1) [token] {$u_1'$};
	\node at (2.9,0.1) [token] {$u_1$};
	\node at (4.5,-0.1) [token] {$u_2'$};
	\node at (7.9,-0.1) [token] {$u_2$};
	\node at (9.5,0.1) [token] {$u_3'$};
	\node at (12.9,0.1) [token] {$u_3$};
	\node at (3.1,4.5) [token] {$v_1'$};
	\node at (3.1,2.7) [token] {$v_1$};
	\node at (9.3,4.5) [token] {$v_4'$};
	\node at (9.3,2.7) [token] {$v_4$};
	\node at (9.2,1.7) [token] {$v_2'$};
	\node at (8.2,0.9) [token] {$v_2$};
	\node at (3.2,1.7) [token] {$v_3'$};
	\node at (4.2,0.9) [token] {$v_3$};
	\node at (7.1,-2.5) [token] {$v_5'$};
	\node at (4.8,-2.3) [token] {$v_5$};
\end{tikzpicture}
\end{center}
\caption{\label{fig:SAT-PTSP}
The instance of Permutation Routing via Matching obtained from the \SAT{} instance $F$ of Example~\ref{ex:SAT-PTSP}.
By moving misplaced tokens along the bold edges, the goal configuration is realized in 3 steps.
The reduction graph described in the proof for Theorem~\ref{thm:PTSPNPhard2} has essentially the same shape.
}
\end{figure}
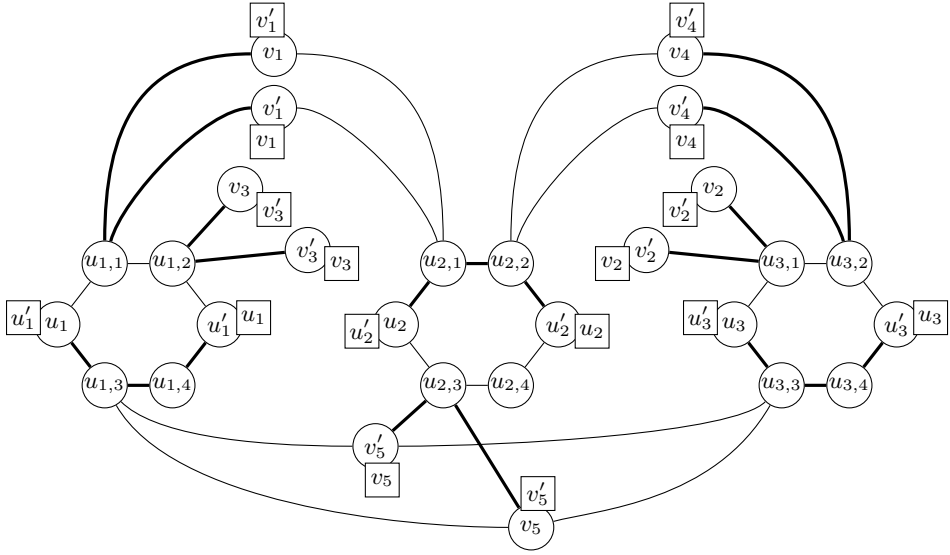
Since $\max\{\,\msf{dist}(w,f^{-1}(w)) \mid w \in V_F \,\} = 3$, $\POPT(G_{F},f) \ge 3$.
We will show that ${F}$ is satisfiable if and only if this lower bound is achieved.
Here we describe an intuition behind the reduction by giving the following observation between a $3$-step solution for $(G_{F},f)$ and  a solution for ${F}$: 
\begin{itemize}
\item tokens $u_i$ and $u_i'$ pass vertices $u_{i,1}$ and $u_{i,2}$ iff $x_i$ should be assigned 0,
	while they pass over $u_{i,3}$ and $u_{i,4}$ iff $x_i$ should be assigned 1,
\item if tokens $v_j$ and $v_j'$ pass a vertex $u_{i,k}$ for some $k \in \{1,2\}$ then $C_j \in {F}_k$ is satisfied thanks to $x_i$, 
	while if they pass over $u_{i,3}$ then $C_j \in {F}_3$ is satisfied thanks to $\neg x_i$.
\end{itemize}
Of course it is contradictory that a clause $C_j \in {F}_1$ is satisfied by $x_i \in C_j$ which is assigned $0$.
This impossibility corresponds to the fact that there are no $i,j$ such that both $u_{i}$ and $v_{j}$ with $C_j \in {F}_1$ go to their respective goals via $u_{i,1}$ in a 3-step solution.
\begin{lemma}\label{lem:3PTSPNPhard}
The formula ${F}$ is satisfiable if and only if\/ $\POPT(G_{F},f) = 3$.
\end{lemma}
\begin{proof}
Suppose that there is $\phi:X \to \{0,1\}$ satisfying ${F}$.
Then each clause must have a literal to which $\phi$ assigns 1.
Let $\psi:{F} \to X$ be such that
 $\psi(C_j) \in C_j$ and $\phi(\psi(C_j))=1$ if $C_j \in {F}_1 \cup {F}_2$,
and $\neg \psi(C_j) \in C_j$ and $\phi(\psi(C_j))=0$ if $C_j \in {F}_3$.
Define
\begin{align*}
S_1 ={}& \{\, \{u_i,u_{i,1}\},\{u_i',u_{i,2}\} \mid \phi(x_i) = 0 \,\}
 \cup  \{\, \{u_i,u_{i,3}\},\{u_i',u_{i,4}\} \mid \phi(x_i) = 1 \,\}
\\	& \cup  \{\, \{v_j,u_{i,k}\} \mid \psi(C_j)=x_i \text{ and }C_j \in {F}_k \,\}\,,
\\
S_2 ={}& \{\, \{u_{i,1},u_{i,2}\} \mid \phi(x_i) = 0 \,\}
	 \cup  \{\, \{u_{i,3},u_{i,4}\} \mid \phi(x_i) = 1 \,\}
\\	& \cup \{\, \{v_{j}',u_{i,k}\} \mid \psi(C_j)=x_i\text{ and }C_j \in {F}_k \,\}\,.
\end{align*}
It is not hard to see that $\lrangle{S_1,S_2,S_1}$ is a solution for $(G_{F},f)$.

Conversely, suppose that  $(G_{F},f)$ admits a solution $\lrangle{S_1,S_2,S_3}$.
Since the token on $u_i$ is moved to $u_i'$ by the three steps, the path that $u_i'$ takes should be either $({u_i,u_{i,1},u_{i,2},u_i'})$ or $({u_i,u_{i,3},u_{i,4},u_i'})$.
In other words, $S_2$ contains at least one of $\{u_{i,1},u_{i,2}\}$ and $\{u_{i,3},u_{i,4}\}$.
We prove that ${F}$ is satisfied by the assignment $\phi:X \to \{0,1\}$ defined as
\[
	\phi(x_i) = \begin{cases}
		0	& \text{ if $\{u_{i,1},u_{i,2}\} \in S_2$\,,}
\\		1	& \text{ otherwise.}
	\end{cases}
\]
For each $C_j \in {F}_1$, the token on $v_j$ must be moved to $v_j'$ via $u_{i,1}$ for some $i$ such that $x_i \in C_j$.
That is, either $\{v_{j},u_{i,1}\} \in S_2$ or $\{v_{j}',u_{i,1}\} \in S_2$.
Since $S_2$ is a parallel swap, $\{u_{i,1},u_{i,2}\} \notin S_2$ in this case, which means $\phi(x_i)=1$.
Hence $C_j$ is satisfied by $\phi$.
Almost the same arguments show that clauses in ${F}_2$ and ${F}_3$ are also satisfied by $\phi$.
\end{proof}
\begin{theorem}\label{thm:3PTSPNPhard}
For any fixed $p \ge 3$, to decide whether\/ $\POPT(G,f) \le p$ is NP-complete even when $G$ is restricted to be a bipartite graph with maximum vertex degree $4$.
\end{theorem}
\begin{proof}
Lemma~\ref{lem:3PTSPNPhard} proves the theorem for $p=3$.
We show the lemma for $p = 3 + h$ with $h > 0$ by inserting paths into appropriate places in the graph constructed above.
We add the following paths of length $h$:
\[
	(u_i,\hat{u}_{i,1},\dots,\hat{u}_{i,h}) \text{ and } (u_i',\hat{u}_{i,1}',\dots,\hat{u}_{i,h}')
\,.\]
The tokens on those paths in the initial configuration will be
\begin{gather*}
	f(u_i)=\hat{u}_{i,1},\ f(\hat{u}_{i,k}) = \hat{u}_{i,k+1} \text{ for } k=1,\dots,h-1,\ f(\hat{u}_{i,h}) = {u}_{i}' \,,
\\
	f(u_i')=\hat{u}_{i,1}',\ f(\hat{u}_{i,k}') = \hat{u}_{i,k+1}' \text{ for } k=1,\dots,h-1,\ f(\hat{u}_{i,h}') = {u}_{i} \,.
\end{gather*}
Then, concerning the vertices where $u_i$ and $u_i'$ are put in the initial and goal configurations, we have $\msf{dist}(u_i,f^{-1}(u_i)) = \msf{dist}(u_i',f^{-1}(u_i')) = p$.
Then either $\{u_{i,1},u_{i,2}\} \in S_{h+2}$ or $\{u_{i,3},u_{i,4}\} \in S_{h+2}$ if $\lrangle{S_1,\dots,S_p}$ is a solution.
Moreover, we replace every edge $\{v_j,u_{i,k}\}$ for $i \in \{1,\dots,m\}$ with $x_i \in C_j \in F_k$ or $\neg x_i \in C_j \in F_k$ by a path $(v_j,\hat{v}_{j,i,1},\dots,\hat{v}_{j,i,h},u_{i,k})$,
where those new vertices $\hat{v}_{j,i,1},\dots,\hat{v}_{j,i,h}$ have the right tokens in the initial configuration, while we keep $f(v_j)=v_j'$ and $f(v_j')=v_j$.
This realizes $\msf{dist}(v_j,f^{-1}(v_j))=\msf{dist}(v_j',f^{-1}(v_j'))=p-1$.
If the token $v_j'$ on the vertex $v_j$ goes to the vertex $v_j'$ via $u_{i,k}$ within $p$ steps, either $\{v_j',u_{i,k}\} \in S_{h+2}$ or $\{\hat{v}_{j,i,h},u_{i,k}\} \in S_{h+2}$ holds.
The same argument in the proof of Lemma~\ref{lem:3PTSPNPhard} works.
\end{proof}
\noindent
Banerjee and Richards~\cite{BanerjeeR17} have shown Theorem~\ref{thm:3PTSPNPhard} using a different reduction.
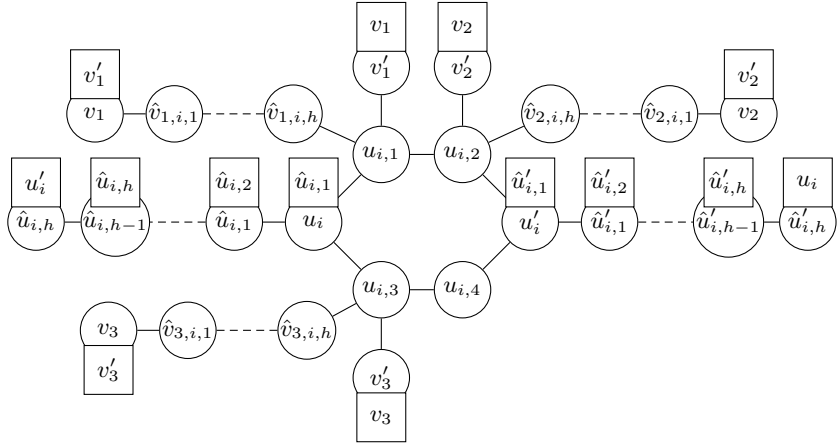
\begin{figure}
\begin{center}
\tikzstyle{vertex}=[circle,draw,inner sep=0pt,minimum size=7.5mm]
\tikzstyle{token}=[rectangle,draw,fill=white,inner sep=0pt,minimum size=6.5mm]
\tikzstyle{dummy}=[rectangle,fill=white,inner sep=0pt,minimum size=3.2mm]
\begin{tikzpicture}[scale=0.9]\small
	\node (u1) at (-0.2,0) [vertex] {$u_i$};
	\node (u11) at (0.8,1.0) [vertex] {$u_{i,1}$};
	\node (u12) at (2.0,1.0) [vertex] {$u_{i,2}$};
	\node (u13) at (0.8,-1.0) [vertex] {$u_{i,3}$};
	\node (u14) at (2.0,-1.0) [vertex] {$u_{i,4}$};
	\node (u1') at (3.0,0) [vertex] {$u_i'$};
	\draw (u1) -- (u11) -- (u12) -- (u1');
	\draw (u1') -- (u14) -- (u13) -- (u1);
	\node (v1) at (-0.5,1.6) [vertex] {$\hat{v}_{1,i,h}$};
	\node (v1') at (0.8,2.3) [vertex] {$v_1'$};
	\draw (v1) -- (u11) -- (v1');
	\node (v2) at (3.3,1.6) [vertex] {$\hat{v}_{2,i,h}$};
	\node (v2') at (2.0,2.3) [vertex] {$v_2'$};
	\draw (v2) -- (u12) -- (v2');
	\node (v3) at (-0.3,-1.6) [vertex] {$\hat{v}_{3,i,h}$};
	\node (v3') at (0.8,-2.3) [vertex] {$v_3'$};
	\draw (v3) -- (u13) -- (v3');
	\node (hu1) [vertex, left of = u1, node distance = 30] {$\hat{u}_{i,1}$};
	\node (hu2) [vertex, left of = hu1, node distance = 45] {$\hat{u}_{i,h-1}$};
	\node (hu3) [vertex, left of = hu2, node distance = 30] {$\hat{u}_{i,h}$};
	\draw (u1) -- (hu1);
	\draw [densely dashed] (hu1) -- (hu2);
	\draw (hu2) -- (hu3);
	\node (hu1') [vertex, right of = u1', node distance = 30] {$\hat{u}_{i,1}'$};
	\node (hu2') [vertex, right of = hu1', node distance = 45] {$\hat{u}_{i,h-1}'$};
	\node (hu3') [vertex, right of = hu2', node distance = 30] {$\hat{u}_{i,h}'$};
	\draw (u1') -- (hu1');
	\draw [densely dashed] (hu1') -- (hu2');
	\draw (hu2') -- (hu3');
	\node (hv11) [vertex, left of = v1, node distance = 45] {$\hat{v}_{1,i,1}$};
	\node (hv12) [vertex, left of = hv11, node distance = 30] {${v}_1$};
	\draw (v1) [densely dashed] -- (hv11);
	\draw (hv11)  -- (hv12);
	\node (hv21) [vertex, right of = v2, node distance = 45] {$\hat{v}_{2,i,1}$};
	\node (hv22) [vertex, right of = hv21, node distance = 30] {${v}_2$};
	\draw (v2) [densely dashed] -- (hv21);
	\draw (hv21)  -- (hv22);
	\node (hv31) [vertex, left of = v3, node distance = 45] {$\hat{v}_{3,i,1}$};
	\node (hv32) [vertex, left of = hv31, node distance = 30] {${v}_3$};
	\draw (v3) [densely dashed] -- (hv31);
	\draw (hv31)  -- (hv32);
	\node [token, above of = u1, node distance = 15] {$\hat{u}_{i,1}$};
	\node [token, above of = hu1, node distance = 15] {$\hat{u}_{i,2}$};
	\node [token, above of = hu2, node distance = 15] {$\hat{u}_{i,h}$};
	\node [token, above of = hu3, node distance = 15] {${u}_{i}'$};
	\node [token, above of = u1', node distance = 15] {$\hat{u}_{i,1}'$};
	\node [token, above of = hu1', node distance = 15] {$\hat{u}_{i,2}'$};
	\node [token, above of = hu2', node distance = 15] {$\hat{u}_{i,h}'$};
	\node [token, above of = hu3', node distance = 15] {${u}_{i}$};
	\node [token, above of = hv12, node distance = 15] {$v_1'$};
	\node [token, above of = v1', node distance = 15] {$v_1$};
	\node [token, above of = hv22, node distance = 15] {$v_2'$};
	\node [token, above of = v2', node distance = 15] {$v_2$};
	\node [token, below of = hv32, node distance = 15] {$v_3'$};
	\node [token, below of = v3', node distance = 15] {$v_3$};
\end{tikzpicture}
\end{center}
\caption{\label{fig:3CPTSP2}
Gadget used to show Theorem~\ref{thm:3PTSPNPhard}, where $x_i \in C_1 \in F_1$, $x_i \in C_2 \in F_2$ and $\neg x_i \in C_3 \in F_3$. Tokens on their goal vertices are omitted.}
\end{figure}

One can modify our reduction so that every vertex has degree at most 3 by dividing vertices $u_{i,k}$ into two vertices of degree at most 3.
Let
\begin{align*}
	V_{F} = {}& \{\, u_i,u_i',u_{i,1},u_{i,1}',u_{i,2},u_{i,2}',u_{i,3},u_{i,3}',u_{i,4},u_{i,4}' \mid 1 \le i \le m\,\}
\\
	& \cup  \{\, v_j,v_j' \mid 1 \le j \le n \,\}
	  \cup \{\, v_{j,i} \mid x_i \in C_j \text{ or } \neg x_i \in C_j \,\}
\,.\end{align*}
The new graph $G_{F}'$ contains the following paths of length $5$ and $4$:
\begin{align*}
	(u_i,u_{i,1}',u_{i,1},u_{i,2}',u_{i,2},u_{i}')	& \text{ and }
	(u_i,u_{i,3}',u_{i,3},u_{i,4}',u_{i,4},u_{i}')	\text{ for each $i \in \{1,\dots,m\}$,}
\\
	(v_j,u_{i,k},u_{i,k}',v_{j,i},v_{j}')	& \text{ if $x_i \in C_j \in {F}_k$ or $\neg x_i \in C_j \in {F}_k$}
\,.\end{align*}
The initial configuration $f$ is defined in the same manner as the previous construction.
It is identity except 
$f(u_i) = u_i'$, 
$f(u_i') = u_i$,
$f(v_j) = v_j'$, and
$f(v_j') = v_j$ for $i \in \{1,\dots,m\}$ and $j \in \{1,\dots, n\}$.
The formula ${F}$ is satisfiable if and only if\/ $\POPT(G_{F}',f) = 5$.
\begin{theorem}\label{thm:PTSPNPhard2}
For any fixed $k \ge 5$, to decide whether\/ $\POPT(G,f) \le k$ is NP-complete even when\/ $G$ is restricted to be a bipartite graph with maximum vertex degree $3$.
\end{theorem}

\subsection{PTIME Subcases}
In this subsection we discuss tractable subcases of Permutation Routing via Matching.
In contrast to Theorem~\ref{thm:3PTSPNPhard}, it is decidable in polynomial time whether an instance of Permutation Routing via Matching admits a 2-step solution.
In addition, we present an approximation algorithm for finding a solution for Permutation Routing via Matching on paths whose length can be at most one larger than that of an optimal solution.

\subsubsection{2-Step Permutation Routing via Matching}
It is well-known that any permutation can be expressed as a product of 2 involutions, which means that any problem instance of Permutation Routing via Matching on a complete graph has a 2-step solution.
Graphs we treat are not necessarily complete but the arguments by Petersen and Tenner~\cite[Lemma~2.3]{PetersenT13} on involution factorization lead to the following observation, which is useful to decide whether $\POPT(G,f) \le 2$ for general graphs $G$.
\begin{proposition}\label{prop:2step}
 $\lrangle{S,T} \in \PSOL(G,f)$ if and only if the set of orbits under $f$ is partitioned as 
$\{\{[u_1]_f,[v_1]_f\},\dots,\{[u_k]_f,[v_k]_f\}\}$ (possibly $[u_j]_f=[v_j]_f$ for some $j \in \{1,\dots,k\}$) so that for every $j \in \{1,\dots,k\}$,
\[
 \{f^i(u_j),f^{-i}(v_j)\} \in \check{S} \text{ and } \{f^{i+1}(u_j),f^{-i}(v_j)\} \in \check{T} \text{ for all $i \in \mbb{Z}$,}
\]
where $\check{S} = S \cup \{\, \{v\} \mid v \in V - \bigcup S\,\}$ for a parallel swap $S$.
\end{proposition}
\begin{theorem}\label{thm:2PTSPPoly}
It is decidable in polynomial time if\/ $\POPT(G,f) \le 2$ for any\/ $G$ and $f$.\footnote{Banerjee and Richards~\cite{BanerjeeR17} independently show this theorem by essentially the same proof.}
\end{theorem}
\begin{proof}
Suppose $G$ and $f$ are given. One can compute in polynomial time all the orbits $[\cdot]_f$.
Let us denote the subgraph of $G$ induced by a vertex set $U \subseteq V$ by $G_{U}$
and the sub-configuration of $f$ restricted to $[u]_f \cup [v]_f$ by $f_{u,v}$.
The set
\begin{align*}
	\Gamma_f = \{\, \{[u]_f,[v]_f\} \mid \POPT(G_{[u]_f \cup [v]_f}, f_{u,v}) \le 2 \,\}
\end{align*}
 can be computed in polynomial time by Proposition~\ref{prop:2step}.
It is clear that $\POPT(G,f) \le 2$ if and only if there is a subset $\Gamma \subseteq \Gamma_f$ in which every orbit occurs exactly once.
This problem is a very minor variant of the problem of finding a perfect matching on a graph, which can be solved in polynomial time~\cite{Edmonds65}.
\end{proof}
One can calculate the number of 2-step solutions in $\PSOL(K_n,f)$ for any configuration $f$ on the complete graph $K_n$ using Petersen and Tenner's formula~\cite{PetersenT13}.
However, it is hard for general graphs.
\begin{theorem}
It is a \#{}P-complete problem to calculate the number of 2-step solutions in $\PSOL(G,f)$ for bipartite graphs $G$.
\end{theorem}
\begin{proof}
We show the theorem by a reduction from the problem of calculating the number of perfect matchings in a bipartite graph $H$, which is known to be \#{}P-complete~\cite{Valiant79}.
For a graph $H=(V,E)$, let the vertex set of $G$ be $V'=\{\, u_i \mid u \in V \text{ and } i \in \{1,2\}\,\}$ and the edge set $E'=\{\, \{u_i,v_j\} \mid \{u,v\} \in E \text{ and } i,j \in \{1,2\} \,\}$. The initial configuration is defined by $f(u_1)=u_2$ and $f(u_2)=u_1$ for all $u \in V$.
If $\lrangle{S,T} \in \PSOL(G,f)$, then for each $u \in V$ there is $v \in V$ such that $\{u,v\} \in E$ and
 either $\{u_1,v_1\},\{u_2,v_2\} \in S$ and $\{u_1,v_2\},\{u_2,v_1\} \in T$ or $\{u_1,v_2\},\{u_2,v_1\} \in S$ and $\{u_1,v_1\},\{u_2,v_2\} \in T$.
Then it is easy to see that $\PSOL(G,f)$ has $2^m$ 2-step solutions if $H$ has $m$ perfect matchings.
Note that if $H$ is bipartite, then so is $G$.
\end{proof}

\subsubsection{Approximation Algorithm for the Permutation Routing via Matching on Paths}
We present an approximation algorithm for the Permutation Routing via Matching on paths which outputs a parallel swap sequence whose length is no more than $\POPT(P_n,f)+1$, where
 $P_n = ( \{\,1,\dots,n\},\{\,\{i,i+1\} \mid 1 \le i < n \,\})$ and $f$ is a configuration on $P_n$.
We say that a swap $\{i,i+1\}$ is \emph{reasonable w.r.t.\ $f$} if $f(i) > f(i+1)$,
and moreover, a parallel swap sequence $\vec{S}=\lrangle{S_1,\dots,S_m}$ is \emph{reasonable w.r.t.\ $f$} if every $e \in S_j$ is reasonable w.r.t.\ $f\lrangle{S_1,\dots,S_{j-1}}$ for all $j \in \{1,\dots,m\}$.
The parallel swap sequence $\lrangle{S_1,\dots,S_m}$ output by Algorithm~\ref{alg:path} is reasonable and 
 satisfies the condition which we call the \emph{odd-even condition}:
 for each odd number $j$, all swaps in $S_j$ are of the form $\{2i-1,2i\}$ for some $i \ge 1$,
and for each even number $j$, all swaps in $S_j$ are of the form $\{2i,2i+1\}$ for some $i \ge 1$.
Our algorithm computes a reasonable odd-even parallel swap sequence in a greedy manner. 
\begin{lemma}\label{lem:maximize}
 Suppose that $g = fS$ for a reasonable parallel swap $S$ w.r.t.\ $f$.
 For any $\lrangle{S_1,\dots,S_m} \in \PSOL(P_n,f)$, there is $\lrangle{S_1',\dots,S_m'} \in \PSOL(P_n,g)$ such that $S_j' \subseteq S_j$ for all $j \in \{1,\dots,m\}$.
\end{lemma}
\begin{proof}
It is enough to show the lemma for the case where $|S|=1$.
Suppose that $S = \{\{i,i+1\}\}$ with $f(i) > f(i+1)$.
By $\lrangle{S_1,\dots,S_m} \in \PSOL(P_n,f)$, at some step we must exchange the positions of the tokens $f(i)$ and $f(i+1)$ in $\lrangle{S_1,\dots,S_m}$.
Let $k$ be the least number such that $\{f_k^{-1}(f(i)),f_k^{-1}(f(i+1))\} \in S_k$ where $f_k = f\lrangle{S_1,\dots,S_k}$.
Define $S_k' = S_k -\{ \{f_k^{-1}(f(i)),\linebreak[0]f_k^{-1}(f(i+1))\}\}$ and $S_j' = S_j$ for all the other $j \in \{1,\dots,m\}-\{k\}$.
Then for any $j \in \{1,\dots,m\}$, $f_j$ and $g_j = g\lrangle{S_1',\dots,S_j'}$ are identical except when $j < k$ the positions of tokens $f(i)$ and $f(i+1)$ are switched.
\end{proof}

\begin{algorithm}[t]
\begin{algorithmic}
	\STATE \textbf{Input}: A configuration $f_0$ on $P_n$
	\STATE \textbf{Output}: A solution $\vec{S} \in \PSOL(P_n,f_0)$
	\STATE Let $j=0$;
	\WHILE{$f_j$ is not identity}
		\STATE Let $j = j+1$, $S_j = \{\,\{i,i+1\}\mid  f_{j-1}(i) > f_{j-1}(i+1) \text{ and $i+j$ is even} \,\}$ and $f_j=f_{j-1} S_j$;
	\ENDWHILE
	\RETURN $\lrangle{S_1,\dots,S_j}$; 
   \end{algorithmic}
   \caption{Approximation algorithm for Permutation Routing via Matching on paths \label{alg:path}}
\end{algorithm}
Let us denote the output of Algorithm~\ref{alg:path} by $\PathAlg(P_n,f_0)$.
Clearly $\PathAlg(P_n,f_0) \in \PSOL(P_n,f_0)$.
\begin{corollary}\label{cor:greedyoe}
For any odd-even solution $\vec{S} \in \PSOL(P_n,f_0)$, we have $|\PathAlg(P_n,f_0)| \le |\vec{S}|$.
\end{corollary}
\begin{proof}
It is obvious that $\PathAlg(P_n,f_0) \in \PSOL(P_n,f_0)$ and it is odd-even.
Suppose that $\vec{S}=\lrangle{S_1,\dots,S_m} \neq \PathAlg(P_n,f_0)$.
Without loss of generality we may assume that $\vec{S}$ is reasonable.
Let $\vec{T}=\lrangle{T_1,\dots,T_k}=\PathAlg(P_n,f_0)$.
If $m \ge k$, we have done.
Suppose $m < k$.
Since the proper prefix $\lrangle{T_1,\dots,T_m}$ of $\vec{T}$ is not a solution, 
there must exist $j \le m$ such that $S_1=T_1,\dots,S_{j-1}=T_{j-1}$ and $S_j \neq T_j$.
Since $\vec{S}$ is reasonable and Algorithm~\ref{alg:path} is greedy, $S_j \subsetneq T_j$ holds.
Applying Lemma~\ref{lem:maximize} to $f_j = f_0\lrangle{S_1,\dots,S_j}$ and $S=T_j-S_j$,
we obtain $S_{j+1}' \subseteq S_{j+1},\dots,S_{m}' \subseteq S_m$ such that $\lrangle{S_1,\dots,S_{j-1},T_j,S_{j+1}',\dots,S_{m}'} \in \PSOL(P_n,f_0)$.
By definition the new solution $\vecp{S}=\lrangle{T_1,\dots,T_{j-1},T_j,S_{j+1}',\dots,S_{m}'}$ is odd-even.
Hence one can apply the same argument to $\vecp{S}$ and finally get $\lrangle{T_1,\dots,T_m} \in \PSOL(P_n,f_0)$.
\end{proof}

\begin{theorem}\label{thm:PTSPPath}
$|\PathAlg(P_n,f_0)| \le \POPT(P_n,f_0)+1$.
\end{theorem}
\begin{proof}

By Corollary~\ref{cor:greedyoe}, it is enough to show that every swap sequence $\vec{S}=\lrangle{S_1,\dots,S_m}$ admits an equivalent odd-even sequence $\vecp{S}$ such that $|\vecp{S}| \le |\vec{S}|+1$.
Without loss of generality we assume that $S_j \cap S_{j+1} = \emptyset$ for any $j$ (in fact, any reasonable parallel swap sequence meets this condition).
For a parallel swap sequence $\vec{S}=\lrangle{S_1,\dots,S_m}$, define $\OES(\vec{S}) = \lrangle{S_1',\dots,S'_{m+1}}$ by delaying swaps which do not meet the odd-even condition, that is,
\[
	S'_j = \{\,  \{i,i+1\} \in S_j \cup S_{j-1} \mid i+j \text{ is even}\,\}
\]
for $j=1,\dots,m+1$ assuming that $S_0=S_{m+1} = \emptyset$.
By the parity restriction, each $S_j'$ is a parallel swap.
It is easy to show by induction on $j$ that
\[
	f\lrangle{S_1',\dots,S_j'}(i) =
	\begin{cases}
		f\lrangle{S_1,\dots,S_{j-1}}(i)	& \text{ if $\{i,i+1\} \in S_j$ and $i+j$ is odd,}
\\		f\lrangle{S_1,\dots,S_{j}}(i)	& \text{ otherwise,}
	\end{cases}
\]
for each $j \in \{1,\dots,m+1\}$, which implies that $f \vec{S} = f \OES(\vec{S})$.
Therefore, for an optimal reasonable solution $\vec{S}_0$, we have
$	|\vec{S}_0| + 1 = |\OES(\vec{S}_0)| \ge |\PathAlg(P_n,f_0)|$.
\end{proof}
\begin{example}
Let us consider the initial configuration $f_0 : \lrangle{3,2,5,1,7,6,4}$ on $P_7$, where we express a configuration $f$ as a sequence $\lrangle{f(1),\dots,f(7)}$.
According to the output by Algorithm~\ref{alg:path}, the configuration changes as follows:
\begin{align*}
f_0: &\ \lrangle{\ul{3,2},\ul{5,1},\ul{7,6},4}\,,
\\
f_1: &\ \lrangle{2,\ul{3,1},5,6,\ul{7,4}}\,,
\\
f_2: &\ \lrangle{\ul{2,1},3,5,\ul{6,4},7}\,,
\\
f_3: &\ \lrangle{1,2,3,\ul{5,4},6,7}\,,
\\
f_4:  &\ \lrangle{1,2,3,4,5,6,7}\,,
\end{align*}
than which an optimal swapping sequence is shorter by one:
\begin{align*}
f_0: &\ \lrangle{\ul{3,2},\ul{5,1},7,\ul{6,4}}\,,
\\
f_1': &\ \lrangle{2,\ul{3,1},5,\ul{7,4},6}\,,
\\
f_2': &\ \lrangle{\ul{2,1},3,\ul{5,4},\ul{7,6}}\,,
\\
f_3': &\ \lrangle{1,2,3,4,5,6,7}\,.
\end{align*}
\end{example}

\section{Coloring Routing via Matching}
\emph{Colored Token Swapping} is a generalization of Token Swapping, where each token is colored and different tokens may have the same color.
By swapping tokens on adjacent vertices, the goal coloring configuration should be realized.
More formally, a \emph{coloring} is a map $f$ from $V$ to $\mbb{N}$. 
The definition of a swap application to a configuration can be applied to colorings with no change.
We say that two colorings $f$ and $g$ are \emph{consistent} if $|f^{-1}(i)|=|g^{-1}(i)|$ for all $i \in \mbb{N}$.
Since the problem is a generalization of Token Swapping, obviously it is NP-hard.
Yamanaka et al.~\cite{YamanakaHKOSUU15} have investigated subcases of Colored Token Swapping called $c$-Colored Token Swapping where the codomain of colorings is restricted to $\{1,\dots,c\}$.
Along this line, we discuss the colored version of Permutation Routing via Matching in this section.

\problemdef{$c$-Coloring Routing via Matching}%
{A graph $G$, two consistent $c$-colorings $f$ and $g$, and a number $k \in \mbb{N}$.}%
{Is there $\vec{S}$ with $|\vec{S}| \le k$ such that $f\vec{S}=g$?}

Define $\POPT(G,f,g) = \min\{\, |\vec{S}| \mid f\vec{S}=g\,\}$ for two consistent colorings $f$ and $g$.
Since $\POPT(G,f,g)$ can be bounded by $\POPT(G,h)$ for some configuration $h$, the $c$-Coloring Routing via Matching belongs to NP.

\subsection{Hardness of the $c$-Coloring Routing via Matching}

Yamanaka et al.\ have shown that the $3$-colored Token Swapping is NP-hard by a reduction from the 3DM.
It is not hard to see that their reduction works to prove the NP-hardness of the $3$-Coloring Routing via Matching.
We then obtain the following theorem as a corollary to their discussion.
\begin{theorem}\label{thm:PCTSPNPhard}
To decide whether\/ $\POPT(G,f,g) \le 3$ is NP-hard even if\/ $G$ is restricted to be a planar bipartite graph with maximum vertex degree $3$ and $f$ and $g$ are $3$-colorings.
\end{theorem}
Yamanaka et al.\ have shown that $2$-Colored Token Swapping is solvable in polynomial time on the other hand.
In contrast, we prove that the $2$-Coloring Routing via Matching is still NP-hard.
\begin{theorem}\label{thm:2PCTSP}
To decide whether\/ $\POPT(G,f,g) \le 3$ is NP-hard 
for a bipartite graph $G$ with maximum vertex degree $4$ and $2$-colorings $f$ and $g$.
\end{theorem}
\begin{proof}
We prove the theorem by a reduction from \SAT{}.
We use a slight modification of the graph used in the proof of Lemma~\ref{lem:3PTSPNPhard} to show the theorem for $p=3$.
We now define
\begin{align*}
	V_F = {}& \{\, u_i,u_i',u_{i,1},u_{i,2},u_{i,3},u_{i,4} \mid 1 \le i \le m\,\}
\\ &	\cup \{\, v_j,v_j' \mid 1 \le j \le n \,\} \cup \{\, v_{j,i} \mid x_i \in C_j \text{ or } \neg x_i \in C_j \,\}
\,.\end{align*}
The edge set $E_{F}$ is the least set that makes $G_{F}$ contain the following paths of length $3$:
\begin{gather*}
	(u_i,u_{i,1},u_{i,2},u_{i}')	
	\text{ and }
	(u_i,u_{i,3},u_{i,4},u_{i}')	 \text{ for each $i \in \{1,\dots,m\}$,}
\\
	(v_j,v_{j,i},u_{i,k},v_{j}')	 \text{ if $x_i \in C_j \in {F}_k$ or $\neg x_i \in C_j \in {F}_k$}
\,.\end{gather*}

The initial and goal colorings $f$ and $g$ are defined to be $f(w)=1$ and $g(w)=1$ for all $w$ but
$	f(u_i) = g(u_i') = 2$ for each $x_i \in X$,
$	f(v_j) = g(v_j') = 2$ for each $C_j \in {F}_1 \cup {F}_3$
and
$ f(v_j') = g(v_j) = 2 $ for each $C_j \in {F}_2$.
Figure~\ref{fig:CPTSP} illustrates the gadget related to a variable $x_1$ that occurs positively in $C_1 \in F_1$, $C_2 \in F_2$ and negatively in $C_3 \in F_3$,
where each vertex $w$ with $f(w)=2$ has a black box on it and one with $g(w)=2$ is represented with a bold rim.
\begin{figure}
\begin{center}
\tikzstyle{vertex}=[circle,draw,inner sep=0pt,minimum size=6mm]
\tikzstyle{gvertex}=[ultra thick, circle,draw,inner sep=0pt,minimum size=6mm]
\tikzstyle{token}=[rectangle,draw,fill=black,inner sep=0pt,minimum size=3.2mm]
\tikzstyle{dummy}=[rectangle,fill=white,inner sep=0pt,minimum size=3.2mm]
\begin{tikzpicture}[scale=0.9]\small
	\node (u1) at (0,0) [vertex] {$u_1$};
	\node (u11) at (0.7,0.9) [vertex] {$u_{1,1}$};
	\node (u12) at (1.7,0.9) [vertex] {$u_{1,2}$};
	\node (u13) at (0.7,-0.9) [vertex] {$u_{1,3}$};
	\node (u14) at (1.7,-0.9) [vertex] {$u_{1,4}$};
	\node (u1-) at (2.4,0) [gvertex] {$u_1'$};
	\draw (u1) -- (u11) -- (u12) -- (u1-);
	\draw (u1-) -- (u14) -- (u13) -- (u1);
	\node (v1) at (-0.6,2.2) [vertex] {$v_1$};
	\node (v1') at (0.7,2.0) [gvertex] {$v_{1}'$};
	\node (v11) at (-0.6,1.1) [vertex] {$v_{1,1}$};
	\draw (v1) -- (v11) -- (u11) -- (v1');
	\node (v2) at (3.0,2.2) [gvertex] {$v_2$};
	\node (v21) at (3.0,1.1) [vertex] {$v_{2,1}$};
	\node (v2') at (1.7,2.0) [vertex] {$v_2'$};
	\draw (v2) -- (v21) -- (u12) -- (v2');
	\node (v3) at (-0.6,-2.2) [vertex] {$v_3$};
	\node (v31) at (-0.6,-1.1) [vertex] {$v_{3,1}$};
	\node (v3') at (0.7,-2.0) [gvertex] {$v_3'$};
	\draw (v3) -- (v31) -- (u13) -- (v3');
	\node at (-0.4,0.0) [token] { };
	\node at (-1.0,2.2) [token] { };
	\node at (-1.0,-2.2) [token] { };
	\node at (1.4,2.3) [token] { };
	\node at (0,-3.5) [dummy] { };
\end{tikzpicture}
\hspace{0.4cm}
\begin{tikzpicture}[scale=0.9]\small
	\node (u1) at (0,0) [vertex] {$u_i$};
	\node (u11') at (0.7,0.9) [vertex] {$ $};
	\node (u11) at (1.7,0.9) [vertex] {$ $};
	\node (u15') at (2.7,0.9) [vertex] {$u_{i,1}$};
	\node (u15) at (3.7,0.9) [gvertex] {$u_{i,2}$};
	\node (u12') at (4.7,0.9) [vertex] {$ $};
	\node (u12) at (5.7,0.9) [vertex] {$ $};
	\node (u13') at (0.7,-0.9) [vertex] {$ $};
	\node (u13) at (1.7,-0.9) [vertex] {$ $};
	\node (u16') at (2.7,-0.9) [vertex] {$u_{i,3}$};
	\node (u16) at (3.7,-0.9) [gvertex] {$u_{i,4}$};
	\node (u14') at (4.7,-0.9) [vertex] {$ $};
	\node (u14) at (5.7,-0.9) [vertex] {$ $};
	\node (u1') at (6.4,0) [gvertex] {$u_i'$};
	\draw (u1) -- (u11') -- (u11) -- (u15') -- node[below] {1} (u15) -- (u12') -- (u12) -- (u1');
	\draw (u1) -- node[above right] {1} (u13') -- node[above] {2} (u13) -- node[above] {3} (u16') -- node[above] {4} node[below] {1} (u16) -- node[below] {2} (u14') -- node[below] {3} (u14) -- node[below right] {4} (u1');
	\node (v1') at (0.7,3.1) [vertex] {$v_1$};
	\node (v11) at (0.7,2.0) [vertex] {$ $};
	\node (v1) at (1.7,2.0) [gvertex] {$v_1'$};
	\draw (v1') -- node[left] {1} (v11) -- node[left] {2} (u11') --  node[below] {3} (u11) --  node[right] {4} (v1);
	\node (v2') at (4.7,2.0) [vertex] {$v_2$};
	\node (v21') at (5.7,2.0) [vertex] {$ $};
	\node (v2) at (5.7,3.1) [gvertex] {$v_2'$};
	\draw (v2') -- node[left] {1} (u12') -- node[below] {2} (u12) -- node[right] {3} (v21') -- node[right] {4} (v2);
	\node (v3') at (0.7,-3.1) [vertex] {$v_3$};
	\node (v31) at (0.7,-2.0) [vertex] {$ $};
	\node (v3) at (1.7,-2.0) [gvertex] {$v_3'$};
	\draw (v3') -- (v31) -- (u13') -- (u13) -- (v3);
	\node at (-0.4,0.0) [token] { };
	\node at (0.5,3.5) [token] { };
	\node at (0.5,-3.5) [token] { };
	\node at (2.5,1.3) [token] { };
	\node at (2.5,-1.3) [token] { };
	\node at (4.5,2.4) [token] { };
\end{tikzpicture}
\end{center}
\caption{\label{fig:CPTSP}
Gadgets used to show Theorems~\ref{thm:2PCTSP} (left) and~\ref{thm:2PCTSP+} (right). We must convey all black boxes to marked vertices via matching.}
\end{figure}
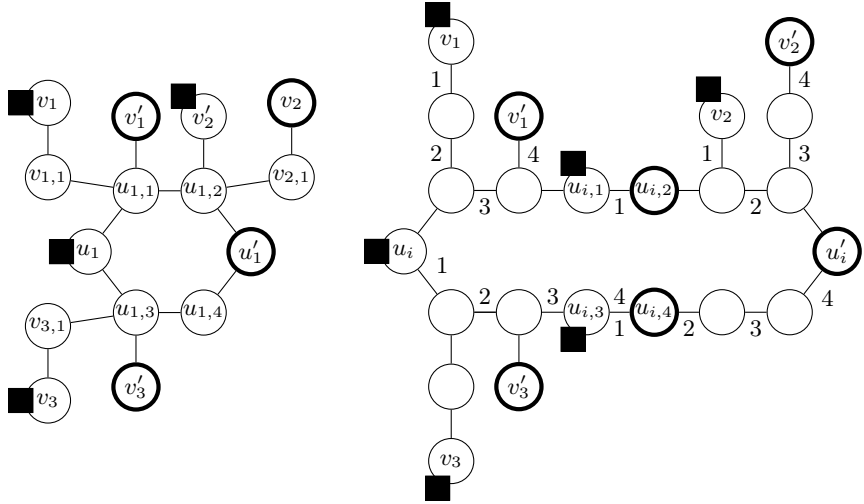

If ${F}$ is satisfiable, then essentially the same parallel swap sequence in the proof of Lemma~\ref{lem:3PTSPNPhard} witnesses $\POPT(G_{F},f,g) = 3$.
It is enough to show that
 if $g = f\vec{S}$ with $|\vec{S}| \le 3$, then the token colored 2 on $u_i$ is moved to $u_i'$ for each $x_i \in X$, the one on $v_j$ is moved to $v_j'$ for $C_j \in {F}_1 \cup {F}_3$, and the one on $v_j'$ is moved to $v_j$ for $C_j \in {F}_2$.
The token on $v_j$ must go to a vertex $w$ such that $g(w)=2$ and $\msf{dist}(v_j,w) \le 3$.
For $C_j \in {F}_1 \cup {F}_3$, the only vertex that meets the condition is $v_j'$.
On the other hand, for each $C_j \in {F}_2$, the vertex $v_j$ requires a token colored $2$ moved from somewhere $w$, i.e., $f(w)=2$ and $\msf{dist}(v_j,w) \le 3$.
The only possibility is the vertex $v_j'$.
Therefore, the token on $u_i$ for $i \in\{1,\dots,m\}$ can be moved to neither $v_j'$ nor $v_j$ for any $j \in \{1,\dots,n\}$.
The unique possible destination of $u_i$ is $u_i'$.
\end{proof}

We can also show the following using the ideas for proving Theorems~\ref{thm:PTSPNPhard2} and~\ref{thm:2PCTSP}.
\begin{theorem}\label{thm:2PCTSP+}
For any fixed $p \ge 4$, to decide whether\/ $\POPT(G,f,g) \le p$ is NP-hard even if\/ $G$ is a bipartite graph with maximum vertex degree $3$ and $f$ and $g$ are $2$-colorings.
\end{theorem}
\begin{proof}
The theorem is shown based on the reduction from \SAT{} used for Theorems~\ref{thm:3PTSPNPhard},~\ref{thm:PTSPNPhard2} and~\ref{thm:2PCTSP} again.
The gadget we use for this theorem is shown on the right in Figure~\ref{fig:CPTSP} for $p=4$.
Suppose that $F$ is satisfied by some $\phi$.
If $\phi(x_i) = 0$, then the vertex on $u_i$, $u_{i,1}$ and $u_{i,3}$ will go to $u_{i,2}$, $u_{i}'$ and $u_{1,4}$, respectively.
Otherwise, they will go to $u_{i,4}$, $u_{i,2}$ and $u_{i}'$.
Then each $v_j$ can be moved to $v_j'$ within $4$ steps using an edge on the $u_i$-$u_i'$ path if either $x_i \in C_j \in F_1 \cup F_2$ and $\phi(x_i)=1$ or $\neg x_i \in C_j \in F_3$ and $\phi(x_i)=0$.
Figure~\ref{fig:CPTSP} illustrates the case where $\phi(x_i)=1$, $x_i \in C_1 \in F_1$, $x_i \in C_2 \in F_2$ and $\neg x_i \in C_3 \in F_3$.
Numbers labeling edges show when they are used in a 4-step solution $\lrangle{S_1,\dots,S_4}$.

On the other hand, suppose that $\POPT(G,f,g) \le 4$.
Considering the destination of the token on the vertex $v_j$ for $C_j \in F_1 \cup F_3$, the unique vertex $w$ such that $f(v_j)=g(w)$ and $\msf{dist}(v_j,w) \le 4$ is $w = v_j'$.
Similarly, considering the vertex $v_j$ for $C_j \in F_2$, the unique vertex $w$ such that $g(v_j')=f(w)$ and $\msf{dist}(v_j',w) \le 4$ is $w = v_j$.
Therefore, all the tokens on $v_j$ are moved to the vertex $v_j'$.
This means that the only possible destinations of the token on $u_i$ are $u_{i,2}$ and $u_{i,4}$.
If $u_i$ is moved to $u_{i,2}$, then the only possible destination of the token on $u_{i,1}$ is $u_{i}'$, and thus the token on $u_{i,3}$ must go to $u_{i,4}$.
It is now clear that $F$ is satisfied by $\phi$ such that $\phi(x_i)=0$ if and only if the token on $u_i$ goes to $u_{i,2}$.

The theorem for $p > 4$ can be shown by inserting extra paths into appropriate places.
\end{proof}

The next theorem contrasts the result on Theorem~\ref{thm:2PTSPPoly}, which shows that it is polynomial-time decidable whether 2-step solution exists in Permutation Routing.
\begin{theorem}\label{thm:3PCTSP2}
It is NP-hard to decide whether $\POPT(G,f,g) \le 2$ for a graph $G$ of vertex degree at most $4$ and $3$-colorings $f$ and $g$.
\end{theorem}
\begin{proof}
We use the graph used in the proof of Lemma~\ref{lem:3PTSPNPhard}.
Recall that we have defined $G_F = (V_F,E_F)$ by
\begin{align*}
	V_{F} = {}& \{\, u_i,u_i',u_{i,1},u_{i,2},u_{i,3},u_{i,4} \mid 1 \le i \le m\,\} \cup \{\, v_j,v_j' \mid 1 \le j \le n \,\}
\,,\end{align*}
and $E_{F}$ containing the following paths
\begin{gather*}
	(u_i,u_{i,1},u_{i,2},u_{i}')	
	\text{ and }
	(u_i,u_{i,3},u_{i,4},u_{i}')	 \text{ for each $i \in \{1,\dots,m\}$,}
\\
	(v_j,u_{i,k},v_{j}')	 \text{ if $x_i \in C_j \in {F}_k$ or $\neg x_i \in C_j \in {F}_k$}
\,.\end{gather*}
The initial and goal 3-colorings $f$ and $g$ are such that
\[
	\begin{array}{c|cccccccccc}
		& u_i	& u_i'	&  u_{i,1}	&  u_{i,2}	&  u_{i,3}	&  u_{i,4}	& v_j	& v_j'	& v_k	& v_k'
	\\ \hline
	f	&  2	&  1	&  	   1	&     	2	&  	   1	&     	2	&  3	&  1	&  3	&	2	
	\\
	g	&  1	&  2	&  	   1	&     	2	&  	   1	&     	2	&  1	&  3	&  2	&	3	
	\end{array}
\]
for $i \in \{1,\dots,m\}$, $C_j \in F_1 \cup F_3$ and $C_k \in F_2$.
Figure~\ref{fig:3CPTSP2} illustrates the reduction where the values of $f$ and $g$ are shown in rectangles and circles, respectively.
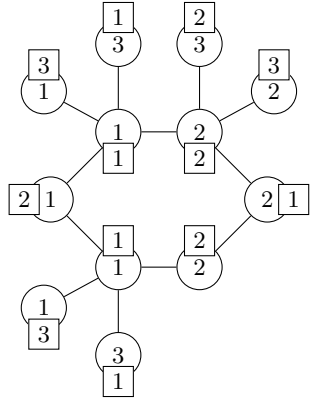
\begin{figure}
\begin{center}
\tikzstyle{vertex}=[circle,draw,inner sep=0pt,minimum size=6mm]
\tikzstyle{token}=[rectangle,draw,fill=white,inner sep=0pt,minimum size=4mm]
\tikzstyle{dummy}=[rectangle,fill=white,inner sep=0pt,minimum size=3.2mm]
\begin{tikzpicture}[scale=0.9]\small
	\node (u1) at (-0.2,0) [vertex] {$1$};
	\node (u11) at (0.8,1.0) [vertex] {$1$};
	\node (u12) at (2.0,1.0) [vertex] {$2$};
	\node (u13) at (0.8,-1.0) [vertex] {$1$};
	\node (u14) at (2.0,-1.0) [vertex] {$2$};
	\node (u1-) at (3.0,0) [vertex] {$2$};
	\draw (u1) -- (u11) -- (u12) -- (u1-);
	\draw (u1-) -- (u14) -- (u13) -- (u1);
	\node (v1) at (-0.3,1.6) [vertex] {$1$};
	\node (v1') at (0.8,2.3) [vertex] {$3$};
	\draw (v1) -- (u11) -- (v1');
	\node (v2) at (3.1,1.6) [vertex] {$2$};
	\node (v2') at (2.0,2.3) [vertex] {$3$};
	\draw (v2) -- (u12) -- (v2');
	\node (v3) at (-0.3,-1.6) [vertex] {$1$};
	\node (v3') at (0.8,-2.3) [vertex] {$3$};
	\draw (v3) -- (u13) -- (v3');
	\node [token, left of = u1, node distance = 10] {$2$};
	\node [token, below of = u11, node distance = 10] {$1$};
	\node [token, below of = u12, node distance = 10] {$2$};
	\node [token, above of = u13, node distance = 10] {$1$};
	\node [token, above of = u14, node distance = 10] {$2$};
	\node [token, right of = u1-, node distance = 10] {$1$};
	\node [token, above of = v1, node distance = 10] {$3$};
	\node [token, above of = v1', node distance = 10] {$1$};
	\node [token, above of = v2, node distance = 10] {$3$};
	\node [token, above of = v2', node distance = 10] {$2$};
	\node [token, below of = v3, node distance = 10] {$3$};
	\node [token, below of = v3', node distance = 10] {$1$};
\end{tikzpicture}
\end{center}
\caption{\label{fig:3CPTSP2}
Gadget used to show Theorem~\ref{thm:3PCTSP2}}
\end{figure}

Suppose that $F$ is satisfied by an assignment $\phi$.
There is a function $\psi:{F} \to X$ such that
 $\psi(C_j) \in C_j$ and $\phi(\psi(C_j))=1$ if $C_j \in {F}_1 \cup {F}_2$,
and $\neg \psi(C_j) \in C_j$ and $\phi(\psi(C_j))=0$ if $C_j \in {F}_3$.
Define
\begin{align*}
S_1 ={}& \{\, \{u_i,u_{i,1}\},\{u_i',u_{i,2}\} \mid \phi(x_i) = 0 \,\}
 \cup  \{\, \{u_i,u_{i,3}\},\{u_i',u_{i,4}\} \mid \phi(x_i) = 1 \,\}
\\	& \cup  \{\, \{v_j,u_{i,k}\} \mid \psi(C_j)=x_i \text{ and }C_j \in {F}_k \,\}\,,
\\
S_2 ={}& \{\, \{u_{i,1},u_{i,2}\} \mid \phi(x_i) = 0 \,\}
	 \cup  \{\, \{u_{i,3},u_{i,4}\} \mid \phi(x_i) = 1 \,\}
\\	& \cup \{\, \{v_{j}',u_{i,k}\} \mid \psi(C_j)=x_i\text{ and }C_j \in {F}_k \,\}\,.
\end{align*}
It is easy to see that $g=f\lrangle{S_1,S_2}$.

We now suppose the converse, $\lrangle{S_1,S_2} \in \PSOL(G_F,f,g)$.
For each $u_i$, the only vertices $w$ such that $\msf{dist}(u_i,w) \le 2$ and $f(u_i)=g(w)=2$ are $u_{i,2}$ and $u_{i,4}$.
The only $u_i$-$u_{i,k}$-path of length at most 2 is $(u_i,u_{i,k-1},u_{i,k})$ for $k \in \{2,4\}$.
Thus, either $\{ u_i,u_{i,1} \} \in S_1$ and $\{u_{i,1},u_{i,2}\} \in S_2$ or $\{u_i,u_{i,3}\} \in S_1$ and $\{u_{i,3},u_{i,4}\} \in S_2$.
We will show that $\phi$ defined by
\[
	\phi(x_i) = \begin{cases}
		0	& \text{ if $\{u_i,u_{i,1}\} \in S_1$,}
		\\
		1	& \text{ if $\{u_i,u_{i,3}\} \in S_1$}
	\end{cases}
\]
satisfies $F$.
Each $v_j$ with $j \in \{1,\dots,n\}$ has only one vertex $w$ such that $\msf{dist}(v_j,w) \le 2$ and $f(v_j)=g(w)=3$, which is $w = v_j'$.
The paths of length at most 2 between $v_j$ and $v_j'$ are of the form $(v_j,u_{i,k},v_j')$ for some $i$ and $k$ such that $x_i \in C_j \in F_k$ or $\neg x_i \in C_j \in F_k$.
For those $i$ and $k$, $\{v_j,u_{i,k}\} \in S_1$ and $\{v_j',u_{i,k}\} \in S_2$ holds.
Suppose $C_j \in F_1$.
In this case, $\{u_{i},u_{i,1}\} \notin S_1$, which implies $\phi(x_i) = 1$. By $x_i \in C_j$, $C_j$ is satisfied.
Suppose $C_j \in F_2$.
Then  $\{u_{i,1},u_{i,2}\} \notin S_2$, which implies $\{u_{i},u_{i,1}\} \notin S_1$ and $\phi(x_i) = 1$. By $x_i \in C_j$, $C_j$ is satisfied.
Suppose $C_j \in F_3$.
In this case, $\{u_{i},u_{i,3}\} \notin S_1$, which implies $\phi(x_i) = 0$. By $\neg x_i \in C_j$, $C_j$ is satisfied.
\end{proof}

\subsection{2-Step 2-Coloring Routing via Matching Is Easy}
In the previous subsection we have shown that $c$-Coloring Routing via Matching is hard even to decide whether a $p$-step solution exists if $c \ge 3$ and $p \ge 2$ or $c \ge 2$ and $p \ge 3$.
We will show that it is easy if $c,p \le 2$.
If $c=1$, we have nothing to do. If $p=1$, the problem is reduced to perfect matching, which can be solved in polynomial time~\cite{Edmonds65}.
Suppose that $\lrangle{S_1,S_2}$ is a 2-step solution for $({G,f,g})$ where $f$ and $g$ are consistent 2-colorings on $G = ({V,E})$.
We say that a swap $\{u,v\}$ is \emph{vacuous} for $f$ if $f(u)=f(v)$.
\begin{lemma}\label{lem:22CRM1}
If $(G,f,g)$ admits a 2-step solution, then there is $\lrangle{S_1,S_2} \in \PSOL(G,f,g)$ such that
\begin{itemize}
	\item $S_1 \cap S_2 = \varnothing$,
	\item no swaps in $S_1$ and in $S_2$ are vacuous for $f$ and for $fS_1$, respectively,
	\item $S_1 \cup S_2$ gives a path matching in $G$.
\end{itemize}
\end{lemma}
\begin{proof}
The first two items are trivial.
Assuming $\lrangle{S_1,S_2}$ satisfies the first two, we show the last.
If $\{u,v\},\{v,w\} \in S_1 \cup S_2$ with $v \neq w$, then either $\{u,v\} \in S_1$ and  $\{v,w\} \in S_2$ or $\{u,v\} \in S_2$ and $\{v,w\} \in S_1$.
This implies that $G'=({V,S_1 \cup S_2})$ has degree bound $2$.
Moreover, if $G'$ has a cycle, then the size must be even.
We show that if $G'$ contains a cycle $(u_1,v_1,u_2,\dots,u_{n},v_{n},u_1)$, then $f(u_i)=g(u_i)$ and $f(v_i)=g(v_i)$ for all $i$.
That is, those edges in the cycle can be removed from $S_1$ and $S_2$. 
Hereafter, by $u_{j}$ we mean $u_i$ such that $1 \le i \le n$ and $i \equiv j$ ({mod} $n$).
Without loss of generality, assume $\{u_i,v_i\} \in S_1$ for all $i$ and $\{v_i,u_{i+1}\} \in S_2$ for all $i$.
Since $\{u_i,v_i\} \in S_1$ is not vacuous for $f$, $f(u_i) \neq f(v_i)$ for all $i$.
Since $\{v_i,u_{i+1}\} \in S_2$ is not vacuous for $fS_1$, $fS_1(v_i) \neq fS_1(u_{i+1}) $, i.e., $f(u_i) \neq f(v_{i+1})$ for all $i$.
Hence, $f(v_i) = f(v_{i+1})$ for all $i$.
Moreover, $g(v_i) = f\lrangle{S_1,S_2}(v_i) = f(v_{i+1})$ for all $i$.
That is, $f(v_i)= g(v_i)$ for all $i$.
Similarly we have $f(u_i)= g(u_i)$ for all $i$.
Those tokens need not be moved at all.
\end{proof}
Hereafter we consider only 2-step solutions that satisfy the condition of Lemma~\ref{lem:22CRM1}.
\begin{lemma}\label{lem:22CRM2}
Let $(u_1,\dots,u_n)$ be a (maximal) path in $G' = ({V,S_1 \cup S_2})$ for a $2$-step solution $\lrangle{S_1,S_2} \in \PSOL(G,f,g)$ satisfying the condition of Lemma~\ref{lem:22CRM1}.
If $n=2$, then $f(u_1)=g(u_2) \neq f(u_2) = g(u_1)$.
If $n \ge 3$,
\begin{itemize}
	\item for all $i \in \{2,\dots,n-2\}$, $f(u_i)=g(u_i) \neq f(u_{i+1}) = g(u_{i+1})$,
	\item if $\{u_1,u_2\} \in S_1$ then $f(u_1) \neq f(u_2)$ and $g(u_1)=g(u_2)$,
	\item if $\{u_{1},u_2\} \in S_2$ then $f(u_1)=f(u_2)$ and $g(u_1) \neq g(u_2)$,
	\item $f(u_1) = g(u_n) \neq g(u_1) = f(u_n)$.
\end{itemize}
\end{lemma}
\begin{proof}
For $n=2$, the lemma holds trivially.
We assume $n \ge 3$.
For readability, we rename each vertex $u_i$ by $i$.
Suppose $\{1,2\} \in S_1$, which implies $\{{2i+1},{2i+2}\} \in S_1$ for $0 \le i \le \floor{(n-2)/2}$ and $\{{2i+2},{2i+3}\} \in S_2$ for $0 \le i \le \floor{(n-3)/2}$.
Since $\{2i+1,2i+2\} \in S_1$ is not vacuous for $f$, $f({2i+1}) \neq f({2i+2})$ for $0 \le i \le \floor{(n-2)/2}$.
Since $\{2i+2,2i+3\} \in S_2$ is not vacuous for $fS_1$, $fS_1({2i+2}) \neq fS_1({2i+3}) $, i.e., $f({2i+1}) \neq f({2i+4})$ for $0 \le i \le \floor{(n-4)/2}$.
That is,
\[
	f(1) \neq f(2) \neq \dots \neq f(n-1)\,. 
\]
By $g({2i+2}) = f\lrangle{S_1,S_2}({2i+2}) = fS_1({2i+3}) =  f({2i+4}) $ for $0 \le i \le \floor{(n-4)/2}$, $g({2i+2}) = f({2i+2})$ for $0 \le i \le \floor{(n-4)/2}$.
By $g({2i+3}) = f\lrangle{S_1,S_2}({2i+3}) = fS_1({2i+2}) = f({2i+1})$ for $0 \le i \le \floor{(n-3)/2}$, $g({2i+3}) = f({2i+1})$ for $0 \le i \le \floor{(n-3)/2}$.
On the other hand, $g(1) = f\lrangle{S_1,S_2}({1}) = fS_1({1}) = f({2}) \neq f(1)$.
Therefore,
\[
	f(1) \neq g(1) = f(2) = g(2) \neq f(3) = g(3) \neq \dots \neq f(n-1) = g(n-1)\,.
\]
We have shown the first and second items of the lemma.
Recall that $\lrangle{S_1,S_2} \in \PSOL(G,f,g)$ implies $\lrangle{S_2,S_1} \in \PSOL(G,g,f)$ and moreover,  $\lrangle{S_2,S_1}$ satisfies the condition of Lemma~\ref{lem:22CRM1}.
This symmetry proves the third.
The fourth is a corollary to those three.
The second and third items imply $f(1) \neq g(1)$, by $f(2)=g(2)$. By the symmetry, $f(n) \neq g(n)$.
Since $f$ and $g$ restricted to the path $\{1,\dots,n\}$ are consistent, it must hold $f(1)=g(n)$ and $f(n)=g(1)$.
\end{proof}
\begin{lemma}\label{lem:22CRM3}
Let $P_n = (\{1,\dots,n\},\,\{\,\{i,i+1\} \mid 1 \le i <n\,\})$.
If 
\begin{itemize}
	\item for all $i \in \{2,\dots,n-2\}$, $f(i)=g(i) \neq f({i+1}) = g({i+1})$,
	\item $f(1) = g(n) \neq g(1) = f(n)$,
\end{itemize}
then $({P_n,f,g})$ admits a 2-step solution.
\end{lemma}
\begin{proof}
Let $S_1 = \{\,\{2i+1,2i+2\} \mid 0 \le i \le \floor{(n-2)/2}\,\}$
and $S_2$ be the rest.
If $f(1) \neq f(2)$, $\lrangle{S_1,S_2}$ is a solution.
If $f(1) = f(2)$, $\lrangle{S_2,S_1}$ is a solution.
\end{proof}

We will reduce the concerned problem to the Vertex-Disjoint Path Problem, which can be solved in polynomial-time~\cite{XieLCZ09}.
\problemdef{Vertex-Disjoint Path Problem}%
{A directed graph $G=(V,E)$, two distinguished vertices $s,t \in V$ and $k \in \mbb{N}$.}%
{Are there $k$ $s$-$t$-paths in $G$ which are vertex-disjoint except $s$ and $t$?}
For a given instance $(G,f,g)$ with $G=(V,E)$ of 2-Coloring Routing via Matching, we give an instance $(H,s,t,k)$ of the Vertex-Disjoint Path Problem as follows.
Let us partition $V$ into
\begin{align*}
V_s &= \{\, u \mid f(u) = 1 \text{ and } g(u) = 2 \,\} \,,
\\
V_t &= \{\, u \mid f(u) = 2 \text{ and } g(u) = 1 \,\} \,,
\\
V_1 & = \{\, u \mid f(u)=g(u)=1 \,\} \,,
\\
V_2 & = \{\, u \mid f(u)=g(u)=2 \,\}
\end{align*}
and define $H = (V', F)$ by $V'=V \cup \{s,t\}$ and
\begin{align*}
F ={} & \{\, (u,v) \in (V_s \times V_t) \cup (V_s \times (V_1 \cup V_2))  \cup ((V_1 \cup V_2) \times V_t)
\\ & {} \cup (V_1 \times V_2) \cup (V_2 \times V_1) \mid \{u,v\} \in E\,\}
 \cup (\{s\} \times V_s)  \cup (V_t \times \{t\})
\,.\end{align*}
\begin{lemma}
$(G,f,g)$ admits a 2-step solution if and only if $(H,s,t)$ admits $|V_s|$ disjoint paths.
\end{lemma}
\begin{proof}
Suppose $\lrangle{S_1,S_2} \in \PSOL(G,f,g)$, which satisfies the condition of Lemma~\ref{lem:22CRM1}.
The graph $G'=(V,S_1 \cup S_2)$ consists of exactly $|V_s|$ disjoint paths by Lemma~\ref{lem:22CRM2}.
Clearly $(H,s,t)$ has corresponding $|V_s|$ disjoint $s$-$t$-paths.

Suppose $(H,s,t)$ admits $|V_s|$ disjoint $s$-$t$-paths.
For each path $(s,u_1,\dots,u_n,t)$, $(u_1,\dots,u_n)$ satisfies the condition of Lemma~\ref{lem:22CRM3}.
Since those paths are disjoint, $(G,f,g)$ admits a $2$-step solution.
\end{proof}
\begin{theorem}\label{thm:22CRM}
It is decidable in polynomial time if\/ $\POPT(G,f,g) \le 2$ for consistent $2$-colorings $f$ and $g$ on a graph $G$.
\end{theorem}


\section*{Acknowledgement}
The authors are deeply grateful to Dana Richards for making us aware of existing studies on Permutation Routing via Matching including their manuscripts.
We also would like to thank the anonymous reviewers of the preliminary version~\cite{KawaharaSY17}.
The work was supported in part by JSPS KAKENHI Grant Number 16K16006, CREST Grant Numbers JPMJCR1401 and JPMJCR1402, Japan.


\newpage

\appendix

\section{Proof that Token Swapping on Lollipop Graphs Is in P}\label{sec:lollipop}
This appendix gives a proof that Algorithm~\ref{alg:lollipop} computes an optimal swapping sequence on lollipop graphs.
We will give an evaluation function on configurations on lollipop graphs $L_{m,n}$ such that
any swap changes the value by one, every swap by the algorithm reduces the value by one, and the value is $0$ if and only if the configuration is the identity.
Algorithm~\ref{alg:lollipop} first moves non-negative tokens to the goal vertices on the path and then moves negative ones in the clique.
The number of swaps needed to move a token $j \in \{0,\dots,n\}$ is evaluated by 
\[
\pi(f,j) = \begin{cases}
j+1	& \text{ if $f^{-1}(j)<0$,}
\\
 \min(j+1,\, \mrm{Inv}(f,j) )
	& \text{ if $f^{-1}(j) \ge 0$,}
 \end{cases}
 \]
 where
 \[
 \mrm{Inv}(f,j) =  |\{\, i \mid i < j \text{ and } f^{-1}(i) > f^{-1}(j)\,\}|
\,. \]
So it takes 
\[
	\pi(f) = \sum_{j = 0}^n \pi(f,j)
\]
swaps to move the non-negative tokens to the goal vertices in total.
We then move the negative tokens in the clique.
For a configuration $f'$ such that $f'(j) = j$ for all $j \ge 0$, the number of swaps needed is
\[
	\nu(f') = m - |\Lambda_{f'}| \text{ where } \Lambda_{f'} = \{\, [i]_{f'} \mid i < 0 \,\}
\,.\text{ (See e.g.~\cite{Jerrum85})}\] 
We need to evaluate $|\Lambda_{f'}|$ for $f'=f\vec{S}$ where $\vec{S}$ moves all the non-negative tokens to their goals.
Let us call an injection $f$ from $\{-m,\dots,k\}$ to $\{-m,\dots,n\}$ for some $k \in \{-1,0,\dots,n\}$ a \emph{pseudo configuration}
 if the range of $f$ includes $\{-m,\dots,-1\}$.
For notational simplicity, a pseudo configuration $f$ will often be identified with the sequence $\lrangle{f(-1),\dots,\linebreak[1]f(-m),f(0),\dots,f(k)}$ or the sequence pair $(\lrangle{f(-1),\dots,f(-m)};\,\lrangle{f(0),\dots,f(k)})$.
For a pseudo configuration $(\vec{i};\, \vec{j})$ where $\vec{i}=\lrangle{ f(-1),\dots,f(-m) }$ and $\vec{j} = \lrangle{ f(0),f(1),\dots,f(k) }$,
we define $\nu$ recursively on $|\vec{j}|$ by
\[
\nu(\vec{i};\vec{j}) =\begin{cases}
m-|\Lambda_f|	& \text{ if $\vec{j}$ is empty,}
\\
 \nu(\vec{i};\lrangle{f(1),\dots,f(k)})	& \text{ if $c < f(0)$,}
\\
 \nu(\vec{i}[f(0)/c];\lrangle{f(1),\dots,f(k)})	& \text{ if $c > f(0)$,}
\end{cases}
\]
where $c = \max (\vec{i})$ and $[a/b]$ replaces $b$ by $a$.
That is, 
\[
f[a/b](i) = \begin{cases} f(i) & \text{if $f(i) \neq b$,}
\\	a & \text{if $f(i) = b$.}
\end{cases}
\]
Note that if $c = \max(\vec{i}) > f(0)$, then $c \ge 0$ and thus $(\vec{i}[f(0)/c];\lrangle{f(1),\dots,f(k)})$ is a pseudo configuration.
Our evaluation function $\Phi$ is given as 
\[
	\Phi(f) = \pi(f) + \nu(f)\,.
\]
Note that $\nu$ is defined on pseudo configurations but $\pi$ and $\Phi$ are defined on (proper) configurations.
It is clear that $\Phi(f) \ge 0$ for any configuration and the equation holds if and only if $f$ is the identity.

\begin{lemma}\label{lem:lp_shorten}
For any $\vec{i},\vec{j}$, there is a sequence $\vecp{i}$ consisting of the $m$ smallest elements from $\vec{i} \cdot \vec{j}$,
where $\vec{i}\cdot\vec{j}$ denotes the concatenation of $\vec{i}$ and $\vec{j}$, such that for any $\vec{k}$
\[
\nu(\vec{i};\,\vec{j}\cdot \vec{k}) = \nu(\vecp{i};\, \vec{k})\,,
\]
provided that $(\vec{i};\,\vec{j}\cdot \vec{k})$ is a pseudo configuration.
\end{lemma}
\begin{proof}
The lemma can be shown by induction on $|\vec{j}|$ just following the definition of $\nu$.
\end{proof}

\begin{lemma}\label{lem:remove}
If\/ $\vec{i}\cdot \vec{j}_1$ contains $m$ or more tokens smaller than $a \ge 0$, then 
\[
	\nu(\vec{i};\, \vec{j}_1 \cdot a \cdot \vec{j}_2) = \nu(\vec{i};\, \vec{j}_1 \cdot  \vec{j}_2)\,,
\]
provided that $(\vec{i};\,\vec{j}_1 \cdot a \cdot \vec{j}_2)$ is a pseudo configuration.
\end{lemma}
\begin{proof}
By induction on $|\vec{j}_1|$.
\end{proof}


Now we are going to prove that any possible swap on the graph changes the value of $\Phi$ by one.
We have three cases depending on where a swap takes place.
First we consider the case where a swap takes inside the clique.
\begin{lemma}\label{lem:lp_negativeswap}
Let $f=(\vec{i};\vec{j})$ and $g=(\vecp{i};\vec{j})$ be pseudo configurations
 such that $\vecp{i}=\vec{i}[a/b,b/a]$ for some distinct tokens $a,b$.
Then
\[
	|\nu(f) - \nu(g)|=1
\,.\]
\end{lemma}
\begin{proof}
We show this by induction on $|\vec{j}|$. 
If $\vec{j}$ is not empty, the claim follows the induction hypothesis immediately.
If $\vec{j}$ is empty, $f$ and $g$ are configurations on the clique of $\{-1,\dots,-m\}$.

\prg{Case 1.} Suppose $[a]_f=[b]_f$.
Let $k = |[a]_f|$ and $b = f^j(a)$.
Then $g^i(a) = f^i(a)$ for $i < j$, $g^j(a) = a$, $g^i(b) = f^{j+i}(a)$ for $i < k-j$ and $g^{k-j}(b) = b$.
That is, $[a]_f=[a]_g \cup [b]_g$, $[a]_g \neq [b]_g$ and $|\Lambda_g| = |\Lambda_f|+1$. Hence $\nu(g)=\nu(f)-1$.

\prg{Case 2.} Suppose $[a]_f \neq [b]_f$.
Let $k_a = |[a]_f|$ and $k_b = |[b]_f|$.
Then $g^i(a) = f^i(a)$ for $i \in \{0,\dots, k_a-1\}$, $g^{k_a+i}(a) = f^i(b)$ for $i \in \{0, \dots, k_b-1\}$
and $g^{k_a+k_b}(a)=a$.
That is, $[a]_g = [b]_g = [a]_f \cup [b]_f$ and  $|\Lambda_g| = |\Lambda_f|-1$.  Hence $\nu(g)=\nu(f)+1$.
\end{proof}

\begin{corollary}\label{cor:lp_negativeswap}
Suppose that $g=fe$ for some swap $e \subseteq \{-1,\dots,-m\}$.
Then $|\Phi(g)-\Phi(f)|=1$.
\end{corollary}
\begin{proof}
Clearly $\pi(f)=\pi(g)$ by definition. The claim follows Lemma~\ref{lem:lp_negativeswap}.
\end{proof}

The following lemma is concerned with the value of $\Phi$ when a swap takes at the joint of the clique and the path.
\begin{lemma}\label{lem:lp_borderswap}
If $g = f \{h,0\}$ with $h < 0$, then $|\Phi(g)-\Phi(f)|=1$.
\end{lemma}
\begin{proof}
We may assume without loss of generality that $h=-1$ for the symmetry.
Let 
\begin{align*}
f &= (a \cdot \vec{i};\, b \cdot \vec{j})\,,
\\
g &= (b \cdot \vec{i};\, a \cdot \vec{j})\,.
\end{align*}
Without loss of generality we assume $a > b$.

\prg{Case 1.}
Suppose $a,b < 0$.
Clearly $\pi(f)=\pi(g)$ and $\max(\vec{i}) \ge 0$.
We have  $|\nu(f)-\nu(g)| = 1$ by applying Lemma~\ref{lem:lp_negativeswap} to the fact
\begin{align*}
\nu(f) &= \nu(a \cdot \vec{i}[b/c];\, \vec{j})\,,
\\
\nu(g) &= \nu(b \cdot \vec{i}[a/c];\, \vec{j})\,,
\end{align*}
where $c = \max(\vec{i})$.

\prg{Case 2.}
Suppose $a \ge 0 > b$.

\prg{Case 2.1.} Suppose $a > \max(\vec{i})$.
We have
\begin{align*}
\nu(f) &=\nu(g)= \nu(b \cdot \vec{i};\, \vec{j})\,.
\end{align*}

All the $m$ elements of $b \cdot \vec{i}$ are smaller than $a$,
which are among $m+a$ tokens smaller than $a$.
Therefore, $\vec{j}$ contains exactly $a$ tokens smaller than $a$, which means $\pi(g,a)=a$.
On the other hand, $\pi(f,a)=a+1$ by definition.
For all other positive tokens $k$, $\pi(f,k)=\pi(g,k)$ holds.

All in all, $\Phi(f)-\Phi(g)=1$.

\prg{Case 2.2.} Suppose $\max(\vec{i}) > a$. Let $c =\max(\vec{i})$.
We have
\begin{align*}
\nu(f) &= \nu(a \cdot \vec{i}[b/c]; \vec{j})\,,
\\
\nu(g) &= \nu(b \cdot \vec{i}[a/c]; \vec{j})\,,
\end{align*}
and $|\nu(f)-\nu(g)|=1$ by Lemma~\ref{lem:lp_negativeswap}.

It remains to show  $\pi(f,k)=\pi(g,k)$ for all positive tokens $k$, which is clear for $k \neq a$.
By definition $\pi(f,a)=a+1$.
The fact $c>a$ implies at most $m-1$ tokens in $b \cdot \vec{i}$ are smaller than $a$,
which are among $m+a$ tokens smaller than $a$.
Hence $\vec{j}$ contains at least $a+1$ tokens smaller than $a$, which means $\pi(g,a)=a+1$.

All in all, $|\Phi(f)-\Phi(g)|=1$.

\prg{Case 3.}
Suppose $a > b \ge 0$.
This case is almost identical to Case 2 except that we need to confirm $\pi(f,b) = \pi(g,b)$ in addition.
The fact $a>b$ implies at most $m-1$ tokens in $a \cdot \vec{i}$ are smaller than $b$,
and $\vec{j}$ contains at least $b+1$ tokens smaller than $b$, which means $\pi(f,b)=\pi(g,b)=b+1$.
\end{proof}

The last case we consider is when a swap takes on the path.
\begin{lemma}\label{lem:lp_positiveswap}
If $g = f\{k,k+1\}$ for some $k \ge 0$,
then $|\Phi(g)-\Phi(f)|=1$.
\end{lemma}
\begin{proof}
Let
\begin{gather*}
f = (\vec{i};\, \vec{j}_1 \cdot \lrangle{a,b}\cdot \vec{j}_2),
\\
g = (\vec{i};\, \vec{j}_1 \cdot \lrangle{b,a}\cdot \vec{j}_2).
\end{gather*}
By Lemma~\ref{lem:lp_shorten}, there exists $\vecp{i}$ consisting of the $m$ smallest tokens from $\vec{i}\cdot\vec{j}_1$ such that
\begin{align*}
	\nu(f) &= \nu(\vecp{i};\, \lrangle{a,b}\cdot \vec{j}_2)\,,
\\	\nu(g) &= \nu(\vecp{i};\, \lrangle{b,a}\cdot \vec{j}_2)\,.
\end{align*}
Without loss of generality we assume $a > b$.

\prg{Case 1.}
Suppose $a,b < 0$.
Clearly $\pi(f)=\pi(g)$.
For the two largest tokens $c$ and $d$ in $\vecp{i}$ with $c > d$, we have
\begin{align*}
\nu(f) &= \nu(\vecp{i}[a/c,b/d]; \vec{j}_2)\,,
\\
\nu(g) &= \nu(\vecp{i}[b/c,a/d]; \vec{j}_2)\,.
\end{align*}
Lemma~\ref{lem:lp_negativeswap} implies $|\Phi(f)-\Phi(g)|=1$.

\prg{Case 2.}
Suppose $a \ge 0$.
We have $\mrm{Inv}(f,a)=\mrm{Inv}(g,a)+1$.

\prg{Case 2.1.}
Suppose $\mrm{Inv}(f,a) \le a$.
In this case, we have $\pi(g,a)=\mrm{Inv}(g,a)=\mrm{Inv}(f,a)-1 = \pi(f,a)-1$ and thus $\pi(f) = \pi(g)+1 $.
The fact that $b \cdot \vec{j}_2$ contains at most $a$ tokens smaller than $a$ implies that
$\vec{i} \cdot \vec{j}_1$ contains at least $m$ tokens smaller than $a$.
That is, all of $\vecp{i}$ are smaller than $a$.
By Lemma~\ref{lem:remove}, we have
\begin{align*}
\nu(f) &= \nu(g) = \nu(\vecp{i};\, {b}\cdot \vec{j}_2) \,.
\end{align*}
All in all, $\Phi(f)=\Phi(g)+1$.

\prg{Case 2.2.}
Suppose $\mrm{Inv}(f,a) = a+1$.
In this case, we have $\pi(f,a)=a+1$, $\pi(g,a)=\mrm{Inv}(g,a)=a$ and thus $\pi(f) = \pi(g)+1$.
The fact that $b \cdot \vec{j}_2$ contains exactly $a+1$ tokens smaller than $a$ implies that
$\vec{i} \cdot \vec{j}_1$ contains exactly $m-1$ tokens smaller than $a$.
That is, all of $\vecp{i}$ are smaller than $a$ except one token $c = \max(\vecp{i})$.
Therefore, 
\begin{align*}
\nu(f) &= \nu(\vecp{i}[a/c];\, b \cdot \vec{j}_2) = \nu(\vecp{i}[b/c];\, \vec{j}_2)\,,
\\
\nu(g) &= \nu(\vecp{i}[b/c];\, a \cdot \vec{j}_2) = \nu(\vecp{i}[b/c];\, \vec{j}_2)\,.
\end{align*}
All in all, $\Phi(f)=\Phi(g)+1$.

\prg{Case 2.3.}
Suppose $\mrm{Inv}(f,a) > a+1$.
In this case, we have  $\pi(f,a)=\pi(g,a)=a+1$ and $\pi(f) = \pi(g)$.
The fact that $b \cdot \vec{j}_2$ contains at least $a+2$ tokens smaller than $a$ implies that
$\vec{i} \cdot \vec{j}_1$ contains at most $m-2$ tokens smaller than $a$.
That is, the two largest tokens $c$ and $d$ in $\vecp{i}$ are bigger than $a$.
Therefore, 
\begin{align*}
\nu(f) &= \nu(\vecp{i}[a/c,b/d]; \vec{j}_2)\,,
\\
\nu(g) &= \nu(\vecp{i}[b/c,a/d]; \vec{j}_2)\,.
\end{align*}
Lemma~\ref{lem:lp_negativeswap} implies $|\nu(f)-\nu(g)|=1$.
All in all, $|\Phi(f)-\Phi(g)|=1$.
\end{proof}

\begin{corollary}\label{cor:lp_lowerbound}
$\Phi(f) \le \OPT(L_{m,n},f)$.
\end{corollary}
\begin{proof}
By Corollary~\ref{cor:lp_negativeswap} and Lemmas~\ref{lem:lp_borderswap} and~\ref{lem:lp_positiveswap}.
\end{proof}

\begin{lemma}
Suppose that our algorithm changes $f$ to $g$ at a point in the run.
Then $\Phi(g) = \Phi(f)-1$.
\end{lemma}
\begin{proof}
Suppose that the algorithm moves a token $a \ge 0$.
If $f^{-1}(a) < 0$ then Case 2.1 of the proof of Lemma~\ref{lem:lp_borderswap} applies and we have $\Phi(f)=\Phi(g)+1$.
If $f^{-1}(a) \ge 0$, the fact that $f(i) < a$ for all $i < 0$ implies that $\mrm{Inv}(f,a) \le a$.
Hence Case 2.1 of the proof of Lemma~\ref{lem:lp_positiveswap} applies and we have $\Phi(f)=\Phi(g)+1$.

Suppose that the algorithm moves a token $a < 0$. Then Case 1 of the proof of Lemma~\ref{lem:lp_negativeswap} applies.
We conclude $\Phi(f)=\Phi(g)+1$.
\end{proof}

Therefore, our algorithm gives a solution of $\Phi(f)$ steps, which is optimal by Corollary~\ref{cor:lp_lowerbound}.
\begin{theorem}\label{thm:lollipop}
Token swapping on lollipop graphs can be solved in polynomial time.
\end{theorem}

\section{Proof that Token Swapping on Star-Path Graphs Is in P}\label{sec:gerbera}
This appendix gives a proof that Algorithm~\ref{alg:gerbera} computes an optimal swapping sequence on star-path graphs $Q_{m,n}$
in a manner similar to Appendix~\ref{sec:lollipop}.
The number of swaps needed to move non-negative tokens to the goal vertices is evaluated by the same function $\pi$.
On the other hand, the number of swaps needed to relocate negative tokens is evaluated differently from the case of lollipop graphs.
The algorithm involves two types of swaps: the ones in the inner \textbf{while} loop and the others.
Let us call the former Type A and the latter Type B.
The negative tokens which must be moved are in $N_f = \{\, f(i) \in \{-m,\dots,-1\,\} \mid f(i) \neq i \,\}$.
Among those, some are on a non-negative vertex and some are on a negative vertex.
Tokens of the former type will be forced to move to $0$ by the moves of non-negative tokens (Type B) and then go to the goal vertex by one step (Type A).
Moves of Type B of those tokens are counted by $\pi$.
On the other hand, tokens $i$ of the latter type form equivalence classes $[i]_f \subseteq N_f$, which require $[i]_f+1$ swaps to be relocated to the goal vertices.
Let
\[
	\Delta_f = \{\, [i]_f \subseteq N_f \mid i < 0 \,\}
\]
and 
\[
	\mu(f) = |N_f|+|\Delta_f|
\,.\]
This value $\mu(f)$ correctly evaluates the number of swaps required to relocate negative tokens in the star graph~\cite{Pak99,YamanakaDIKKOSSUU14}.
One might think $\pi(f)+\mu(f)$ could be the right evaluation for $\OPT(Q_{m,n},f)$.
However, when the vertex 0 is occupied by a negative token $i <0$ and the vertex $i$ is occupied by the positive token $j$ which is the largest among the tokens on negative vertices, then the move of $i$ to $i$ (Type A) causes the right move of $j$ to $0$, which reduces the number of swaps required to move $j$ to the goal.
That is, actually $\pi$ overestimates the number of swaps for $j$.
We must discount the evaluation from $\pi(f)+\mu(f)$.
For a pseudo configuration $f = (\vec{i};\,\vec{j}) = (\lrangle{i_1,\dots,i_m};\,\lrangle{j_1,\dots,j_k})$ and $c = \max(\vec{i})$,
define
\[
 	\delta(\vec{i};\,\vec{j}) = \begin{cases}
	0	&	\text{ if $c < 0$,}
\\	\delta(\vec{i};\, \lrangle{j_2,\dots,j_k})	& \text{ if $j_1 > c \ge 0$,}
\\	\delta(\vec{i}[j_1/c] ;\, \lrangle{j_2,\dots,j_k})	& \text{ if $c > j_1 \ge 0$,}
\\	\delta(\vec{i}[j_1/i_{-j_1}] ;\, \vec{j}[i_{-j_1}/j_1])-1	& \text{ if $j_1 < 0$ and $i_{-j_1} = c$,}
\\	\delta(\vec{i}[j_1/i_{-j_1}] ;\, \vec{j}[i_{-j_1}/j_1])	& \text{ otherwise.}
	\end{cases}
\]
Note that if $j_1 < 0$, then $c \ge 0$.
The discount function $\delta$ is well-defined, since the sum of the number of the misplaced tokens in $\vec{i}$ and the length of $\vec{j}$ decreases by one on the right-hand side in the above definition when $c \ge 0$.

Our evaluation function $\Psi$ is given as 
\[
	\Psi(f) = \pi(f) + \mu(f) + \delta(f)\,.
\]
It is clear that $\Psi(f) = 0$ if $f$ is the identity. 

For a pseudo configuration $(\lrangle{i_1,\dots,i_m};\,a)$,
let us define \[
\gamma(\lrangle{i_1,\dots,i_m};\,a) = \begin{cases}
\gamma(\lrangle{i_1,\dots,i_{-a-1},a,i_{-a+1},\dots,i_m};\,i_{-a})	& \text{ if $a < 0$,}
\\
(\lrangle{i_1,\dots,i_m};\,a)	& \text{ if $a \ge 0$.}
\end{cases}
\]
The function $\gamma$ simulates the \textbf{while} loop of Algorithm~\ref{alg:gerbera} in the sense that
if the algorithm has $(\vec{i};\,a \cdot \vec{j})$ as the value of $f$ at the beginning of the \textbf{while} loop, 
it will be $(\vecp{i};\,a' \cdot \vec{j})$ when exiting the loop for $(\vecp{i};\,a')=\gamma(\vec{i};\,a)$.
\begin{lemma}\label{lem:mucon}
Let $\gamma(\vec{i};\, a) = (\vecp{i};\,b)$. Then
\[
	\delta(\vec{i};\, a\cdot \vec{j}) = \begin{cases}
	\delta(\vecp{i};\,b\cdot \vec{j}) - 1	& \text{ if\/ $b=\max(\vec{i})$,}
\\
	\delta(\vecp{i};\,b\cdot \vec{j}) 	& \text{ otherwise.}
	\end{cases}
\]
\end{lemma}
\begin{proof}
We show the lemma by induction on the definition of $\gamma$.
If $\gamma(\vec{i};\, a) = (\vec{i};\, a)$, $a$ does not occur in $\vec{i}$, so $\delta(\vec{i};\, a) = \delta(\vec{i};\, a)$.
Otherwise, suppose $a < 0$ and $\gamma(\vec{i};\, a) = \gamma(\vec{i}[a/i_{-a}];\, i_{-a})$.
Remember $\max(\vec{i}) \ge 0$.
If $i_{-a} = \max(\vec{i})$, then $\gamma(\vec{i}[a/i_{-a}];\, i_{-a}) = (\vec{i}[a/i_{-a}];\, i_{-a})$ and
\[
\delta(\vec{i};\, a) = \delta(\vec{i}[a/i_{-a}];\, i_{-a}) - 1 
\,.\]
If $i_{-a} < \max(\vec{i})$, then $\delta(\vec{i};\, a) = \delta(\vec{i}[a/i_{-a}];\, i_{-a})$ and  $\gamma(\vec{i};\, a) = \gamma(\vec{i}[a/i_{-a}];\, i_{-a})$.
Since $a < 0 \le \max(\vec{i})$ and $i_{-a} < \max(\vec{i})$, we have $\max(\vec{i}) = \max(\vec{i}[a/i_{-a}])$.
By the induction hypothesis, we obtain the lemma.
\end{proof}

\begin{lemma}\label{lem:sp_shorten}
For any $\vec{i},\vec{j}$, there are an integer $\alpha \le 0$ and a sequence $\vecp{i}$ consisting of the $m$ smallest elements from $\vec{i} \cdot \vec{j}$ such that
 for any $\vec{k}$
\[
\delta(\vec{i};\,\vec{j}\cdot \vec{k}) = \delta(\vecp{i};\, \vec{k}) + \alpha
\]
provided that $(\vec{i};\,\vec{j}\cdot \vec{k})$ is a pseudo configuration.
\end{lemma}
\begin{proof}
This is immediate by the definition of $\delta$.
If the definition derives the equation $\delta(\vec{i};\, \vec{j}) = \delta(\vecp{i};\, \vecp{j}) + \alpha = \alpha$ with $\max(\vecp{i}) < 0$,
then $\delta(\vec{i};\, \vec{j} \cdot \vec{k}) = \delta(\vecp{i};\, \vec{k}) + \alpha$ holds anyway.
\end{proof}

\begin{lemma}\label{lem:sp_remove}
If\/ $\vec{i}\cdot \vec{j}_1$ contains $m$ or more tokens smaller than $k \ge 0$, then 
\[
	\delta(\vec{i};\, \vec{j}_1 \cdot k \cdot \vec{j}_2) = \delta(\vec{i};\, \vec{j}_1 \cdot  \vec{j}_2)
\]
provided that $(\vec{i};\,\vec{j}_1 \cdot k \cdot \vec{j}_2)$ is a pseudo configuration.
\end{lemma}
\begin{proof}
We show the lemma by induction on the definition of $\delta$.
Let $c=\max(\vec{i})$.
If $c < 0$, then $\delta(\vec{i};\, \vec{j}_1 \cdot k \cdot \vec{j}_2) = \delta(\vec{i};\, \vec{j}_1 \cdot  \vec{j}_2) = 0$.
Suppose $c \ge 0$.
If $\vec{j}_1$ is empty, $k > c \ge 0$ by the assumption.
The equation holds immediately by definition.
If $\vec{j}_1$ is not empty, the recursive definition of $\delta$ gives $\vecp{i}$ and $\vecp{j}_1$ such that 
\begin{gather*}
\delta(\vec{i};\, \vec{j}_1 \cdot k \cdot \vec{j}_2) = \delta(\vecp{i};\,\vecp{j}_1 \cdot k \cdot \vec{j}_2 ) + \alpha
\\
 \delta(\vec{i};\, \vec{j}_1 \cdot  \vec{j}_2) = \delta(\vecp{i};\,\vecp{j}_1 \cdot \vec{j}_2 ) + \alpha
\end{gather*}
for some $\alpha \in \{0,-1\}$.
To apply the induction hypothesis, it suffices to show that $\vecp{i} \cdot \vecp{j}_1 $ contains at least $m$ tokens smaller than $k \ge 0$.
The only non-trivial case is that the number of tokens smaller than $k$ in $\vecp{i}\cdot\vecp{j}_1$ is smaller than that in $\vec{i}\cdot \vec{j}_1$.
In such a case, for the first element $j_0$ of $\vec{j}_1$,
either $j_0 > c \ge 0$ and $j_0 < k$ ($j_0$ is absent in $\vecp{i}\cdot\vecp{j}_1$) or $c > j_0 \ge 0$ and $c < k$ ($c$ is absent in $\vecp{i}\cdot\vecp{j}_1$) holds.
In the former case, $k > c$ implies that all the $m$ tokens in $\vecp{i} = \vec{i}$ are smaller than $k$.
In the latter case, $k > c >j_0$ implies all the $m$ tokens in $\vecp{i} = \vec{i}[j_0/c]$ are smaller than $k$.
\end{proof}

\begin{corollary}\label{cor:sp_remove}
For any configuration $f = ( \vec{i};\,\vec{j}_1\cdot k \cdot\vec{j}_2 )$ with $k \ge 0$,
if\/ $\mrm{Inv}(f,k) \le k$ then 
\[
	\delta(\vec{i};\, \vec{j}_1 \cdot k \cdot \vec{j}_2) = \delta(\vec{i};\, \vec{j}_1 \cdot  \vec{j}_2)
\,.\]
In particular if\/ $\vec{i}$ contains negative tokens only, $\delta(f)=0$.
\end{corollary}
\begin{proof}
Recall that there exist just $k+m$ tokens smaller than $k$.
The fact $\mrm{Inv}(f,k) \le k$ means that $\vec{j}_2$ contains at most $k$ tokens smaller than $k$, so $\vec{i}\cdot\vec{j}_1$ must have at least $m$ such tokens.
Lemma~\ref{lem:sp_remove} applies.
\end{proof}

\subsection{$\Psi$ Evaluates Our Algorithm}

\begin{lemma}\label{lem:algorithm}
Suppose that our algorithm changes $f$ to $g$ at a point in the run.
Then $\Psi(g) = \Psi(f)-1$.
\end{lemma}
\begin{proof}
We have two types of swaps.

\prg{Case A.}
When the algorithm moves the token $f(0)<0$ to the vertex $f(0)$ (Type A).
In this case we have $|N_g|=|N_f|-1$.
Let $a=f(0)$ and $b=f(a)$, which implies $g(0)=b$ and $g(a)=a$.
Let $I = \{f(-1),\dots,f(-m)\}$ be the set of tokens on the negative vertices in $f$.

\prg{Case A.1.} $0 \le b = \max I$.
Clearly $\delta(f)=\delta(g)-1$ and $\mu(f)=\mu(g)+1$.
Since $\pi(f,i)=\pi(g,i)$ for all $i \neq b$, it is enough to show $\pi(f,b)=\pi(g,b)+1$.
By definition $\pi(f,b) = b + 1$.
Recall that there are exactly $b+m$ tokens that are smaller than $b$.
Since the $m$ tokens on the negative vertices in $g$ are all smaller than $b$,
there are  exactly $b$ tokens smaller than $b$ on non-negative vertices under $g$.
That is, $\mrm{Inv}(g,b)=b$ and thus $\pi(g,b)=b = \pi(f,b)-1$.

\prg{Case A.2.} $0 \le b < \max I$.
Clearly $\delta(f)=\delta(g)$ and $\mu(f)=\mu(g)+1$.
To see $\pi(f,i)=\pi(g,i)$ for all $i \in \{0,\dots,n\}-\{b\}$ is trivial, so it is enough to show $\pi(f,b)=\pi(g,b)$.
By definition $\pi(f,b) = b + 1$.
Recall that there are exactly $b+m$ tokens that are smaller than $b$, of which at most $m-1$ tokens can be on negative vertices in $g$,
since at least one negative vertex is occupied by a token bigger than $b$.
Therefore, there are at least $b+1$ tokens smaller than $b$ on non-negative vertices in $g$.
That is, $\mrm{Inv}(g,b) \ge b+1$ and thus $\pi(g,b)=b+1 = \pi(f,b)$.

\prg{Case A.3.} $b < 0$.
Clearly $\pi(f,i)=\pi(g,i)$ for all $i \ge 0$ and $\delta(f)=\delta(g)$ by definition.
One can easily see $\Delta_f=\Delta_g$, for $[a]_f=[b]_f \notin \Delta_f$, $[a]_g \notin \Delta_g$ and $[b]_g \notin \Delta_g$.
Hence $\mu(g)=\mu(f)-1$.

\prg{Case B.} When the algorithm moves a token $k$ as a move of Type B.

\prg{Case B.1.} $k \ge 0$ and $f^{-1}(k) < 0$.
Let $a=f^{-1}(k)$ and $f(0)=b$, that is, $g(a)=b$ and $g(0)=k$.
By the behavior of the algorithm, we have $f(i) \le k$ for all $i \le 0$.
Since $b \ge 0$,  we have $\mu(f)=\mu(g)$ and $\delta(f)=\delta(g)$.
It is trivially true that $\pi(f,i)=\pi(g,i)$ for all $i \in \{0,\dots,n\}-\{k,b\}$.
Thus it is enough to show that $\pi(g,k)+\pi(g,b) = \pi(f,k)+\pi(f,b)-1 $.
By definition $\pi(f,k)=k+1$ and $\pi(g,b)=b+1$.
Since all the $m$ tokens on the negative vertices in $g$ are smaller than $k$,
the other $k$ tokens smaller than $k$ are found on some non-negative vertices.
That is, $ \mrm{Inv}(g,k)=k$ and thus
\(
\pi(g,k) = k = \pi(f,k)-1
\).
On the other hand in $f$, at least one token, namely $k$, on a negative vertex is bigger than $b$.
Therefore, at least $b+1$ tokens smaller than $b$ are on some non-negative vertices in $f$.
That is, $ \mrm{Inv}(f,b) \ge b+1$ and thus $\pi(f,b)=b+1$.
Therefore, $\pi(g)=\pi(f)-1$.

\prg{Case B.2.} $k \ge 0$ and $f^{-1}(k) \ge 0$.
Clearly $\mu(g)=\mu(f)$, $\mrm{Inv}(g,k)=\mrm{Inv}(f,k) - 1$ and $\mrm{Inv}(g,j)=\mrm{Inv}(f,j)$ for all $j \in \{0,\dots,n\}-\{k\}$.
By the behavior of the algorithm, $f(j) = j$ for all $j > k$ and thus $\mrm{Inv}(f,k) \le k$ and $\pi(g,k)=\pi(f,k)-1$. 
Hence $\pi(g)=\pi(f)-1$.
Corollary~\ref{cor:sp_remove} implies $\delta(g)=\delta(f)$.

\prg{Case B.3.} $k < 0$.
The case where $f^{-1}(k) = 0$ can be discussed as in Case~A.3.
We assume $f^{-1}(k)<0$, in which case we have $f(i)=i$ for all $i \ge 0$ by the behavior of the algorithm.
Clearly $[k]_f \in \Delta_f$ and $\Delta_g = \Delta_f - \{[k]_f\}$, thus $|\Delta_g| = |\Delta_f| - 1$ and $\mu(g)=\mu(f)-1$.
On the other hand, $\pi(f,0)=0$ and $\pi(g,0)=1$, while $\pi(f,j)=\pi(g,j)$ for all $j > 0$.
We have $\delta(f)=0$ by Corollary~\ref{cor:sp_remove} and
\(
\delta(g) = -1
\)
by Lemma~\ref{lem:mucon}.
\end{proof}

\begin{corollary}
For any configuration $f$, $\Psi(f) \ge 0$.
Moreover, $\Psi(f) = 0$ if and only if $f$ is the identity.
\end{corollary}

\subsection{$\Psi$ Is the Right Evaluation Function}

Now we are going to prove that any possible swap on the graph changes the value of $\Psi$ by one.
We have 6 cases depending on the signs of swapped tokens and the vertices where the swap takes place.
Namely we discuss cases where the tokens are both non-negative (Lemma~\ref{lem:pospos}), where one is non-negative and the other is negative (Lemma~\ref{lem:posneg}) and where both are negative (Lemma~\ref{lem:negneg}).
Each case has two subcases depending on whether one of the tokens is on a negative vertex or not.
Lemmas~\ref{lem:swappositive} and~\ref{lem:resolution} are useful to prove those lemmas.

\begin{lemma}\label{lem:swappositive}
Let $(\vec{i};\vec{j})$ and $(\vecp{i};\vec{j})$ be pseudo configurations
 such that $\vec{i}$ contains two distinct non-negative numbers $a,b \ge 0$
and $\vecp{i} = \vec{i}[a/b,b/a]$.
Then
\[
	|\delta(\vec{i};\, \vec{j} ) - \delta(\vecp{i};\,\vec{j} )|=1
\,.\]
\end{lemma}
Note that $\vec{j}$ cannot be empty, since $\vec{i}\cdot\vec{j}$ is a pseudo configuration.
\begin{proof}
It is enough to show that for any $a, b \ge 0$, $\vec{i},\vec{j}$ and $d$,
\[
	|\delta(\lrangle{a,b} \cdot \vec{i};\, {d} \cdot \vec{j} ) - \delta(\lrangle{b,a} \cdot \vec{i};\,{d} \cdot  \vec{j} )|=1
\,.\]
We show this claim by induction on the definition of $\delta$.
If $d \notin \{-1,-2\}$, the proof is trivial.
For the symmetry, we discuss the case where $d=-1$ only.
Without loss of generality we assume $a < b$.
Let $c = \max(\vec{i})$.

\prg{Case 1.} In the case where $b > c$, we have
\begin{align*}
 \delta(\lrangle{a,b} \cdot \vec{i};\, {-1} \cdot \vec{j} ) 
&= \delta(\lrangle{-1,b} \cdot \vec{i};\, {a} \cdot \vec{j} ) 
= \delta(\lrangle{-1,a} \cdot \vec{i};\, \vec{j} ) \, ,
\\
\delta(\lrangle{b,a} \cdot \vec{i};\, {-1} \cdot \vec{j} ) 
&= \delta(\lrangle{-1,a} \cdot \vec{i};\, \vec{j} ) -1
\,.\end{align*}
The claim holds.

\prg{Case 2.} 
In the case where $b < c$,
\begin{align*}
\delta(\lrangle{a,b} \cdot \vec{i};\, {-1} \cdot \vec{j} ) 
&= \delta(\lrangle{-1,b} \cdot \vec{i};\, {a} \cdot \vec{j} ) 
=  \delta(\lrangle{-1,b} \cdot \vec{i}[a/c];\, \vec{j} ) ,
\\
\delta(\lrangle{b,a} \cdot \vec{i};\, {-1} \cdot \vec{j} ) 
&= \delta(\lrangle{-1,a} \cdot \vec{i};\, {b}\cdot \vec{j} )
= \delta(\lrangle{-1,a} \cdot \vec{i}[b/c];\, \vec{j} ) 
\,.\end{align*}
The claim follows the induction hypothesis.
\end{proof}

Let $f = (\vec{i};\,\vec{j})$ be a pseudo configuration where $\vec{i}$ contains a negative token $a$.
The \emph{$a$-resolution} of $\vec{i}$ is defined by 
\[
\vecp{i} = \vec{i}[a/f(a),f(a)/f^2(a),\dots,f^{k-1}(a)/f^k(a),f^k(a)/a]
\]
where $k$ is the least natural number such that either $f^{k+1}(a)=a$ or $f^{k}(a) \ge 0$.
That is, we relocate tokens $a,f(a),\dots,f^{k-1}(a)$ on negative vertices to their respective goals
and push $f^k(a)$ out to $a$, which is actually its goal if $f^{k+1}(a)=a$. 
We also call $g=(\vecp{i};\,\vec{j})$ the \emph{$a$-resolution} of $f$.
If $\gamma(\vec{i};\,a) = (\vec{j};\,b)$ and $a < 0$, it is easy to see that $\vec{j}$ is the $a$-resolutions of $\vec{i}[a/b]$.
\begin{lemma}\label{lem:resolution}
	If\/ $g$ is the $a$-resolution of\/ $f$,
	then $\delta(f)=\delta(g)$.
\end{lemma}
\begin{proof}
	By induction on the definition of $\delta$.
\end{proof}

\begin{lemma}\label{lem:pospos}
Suppose that $g$ is obtained from $f$ by swapping non-negative tokens.
Then $|\Psi(g)-\Psi(f)|=1$.
\end{lemma}
\begin{proof}
Suppose that non-negative tokens $a$ and $b$ are swapped.
Obviously $\mu(g)=\mu(f)$.
We have two cases depending on where those tokens are swapped.

\prg{Case 1.}
The swap takes place on two non-negative vertices.  That is,
\begin{gather*}
f = (\vec{i};\, \vec{j}_1 \cdot \lrangle{a,b}\cdot \vec{j}_2),
\\
g = (\vec{i};\, \vec{j}_1 \cdot \lrangle{b,a}\cdot \vec{j}_2),
\end{gather*}
for some $a,b \ge 0$.
Without loss of generality we assume $a < b$.
We have $\pi(f,i)=\pi(g,i)$ for all $i \in \{0,\dots,n\}-\{b\}$, while $\mrm{Inv}(g,b)=\mrm{Inv}(f,b)+1$.
By Lemma~\ref{lem:sp_shorten}, there exists $\vecp{i}$ consisting of the $m$ smallest tokens from $\vec{i}\cdot\vec{j}_1$ such that
\begin{align*}
	\delta(f) &= \delta(\vecp{i};\, \lrangle{a,b}\cdot \vec{j}_2) + \alpha,
\\	\delta(g) &= \delta(\vecp{i};\, \lrangle{b,a}\cdot \vec{j}_2) + \alpha
\end{align*}
for some $\alpha \le 0$.

\prg{Case 1.1.}  $\mrm{Inv}(f,b) < b$ and $\mrm{Inv}(g,b) < b+1$.
In this case, we have $\pi(f,b)=\mrm{Inv}(f,b)$, $\pi(g,b)=\mrm{Inv}(g,b)$ and thus $\pi(g) = \pi(f)+1 $.
Corollary~\ref{cor:sp_remove} applies to both $f$ and $g$ and we obtain 
\[ \delta(g) = \delta(f) = \delta(\vec{i};\, \vec{j}_1 \cdot {a}\cdot \vec{j}_2)
\] and $\Psi(g)=\Psi(f)+1$.

\prg{Case 1.2.}  $\mrm{Inv}(f,b) = b$ and $\mrm{Inv}(g,b) = b+1$.
In this case, we have $\pi(g) = \pi(f)+1$.
Since $\vec{j}_2$ contains $b$ tokens smaller than $b$, on the other hand $\vec{i}\cdot\vec{j}_1$ contains exactly $m-1$ tokens smaller than $b$.
Let $c$ be the unique element of $\vecp{i}$ which is bigger than $b$.
Then
\begin{align*}
	\delta(f) &= \delta(\vecp{i}[a/c];\, {b}\cdot \vec{j}) + \alpha = \delta(\vecp{i}[a/c];\, \vec{j}) + \alpha,
\\	\delta(g) &= \delta(\vecp{i}[b/c];\, {a}\cdot \vec{j}) + \alpha = \delta(\vecp{i}[a/c];\, \vec{j}) + \alpha,
\end{align*}
since $b$ is the biggest in $\vecp{i}[b/c]$.
Hence $\delta(g) = \delta(f)$ and $\Psi(g)=\Psi(f)+1$.

\prg{Case 1.3.}  $\mrm{Inv}(f,b) > b$ and $\mrm{Inv}(g,b) > b+1$.
In this case, $\pi(f,b)=\pi(g,b)=b+1$ and $\pi(g) = \pi(f)$.
Since $\vec{j}_2$ contains at least $b+1$ tokens smaller than $b$, $\vec{i}\cdot\vec{j}_1$ contains at most $m-2$ tokens smaller than $b$.
Hence $\vecp{i}$ contains at least $2$ tokens bigger than $b$.
Let $c_1$ and $c_2$ be the biggest and second biggest in $\vecp{i}$, respectively.
Then
\begin{align*}
	\delta(f) &= \delta(\vecp{i}[a/c_1,b/c_2];\, {b}\cdot \vec{j}) + \alpha ,
\\	\delta(g) &= \delta(\vecp{i}[b/c_1,a/c_2];\, {a}\cdot \vec{j}) + \alpha .
\end{align*}
By Lemma~\ref{lem:swappositive}, we obtain  $|\delta(g) - \delta(f)|=1$ and $|\Psi(g)-\Psi(f)|=1$.

\prg{Case 2.}
The swap takes place between a negative vertex and $0$.
Without loss of generality we may assume that the negative position is $-1$ and $f(-1)<g(-1)$.
That is,
\begin{gather*}
f = (a \cdot \vec{i};\, b \cdot \vec{j}),
\\
g = (b \cdot \vec{i};\, a \cdot \vec{j}),
\end{gather*}
where $0 \le a<b$.
By definition $\pi(f,a)=a+1$ and $\pi(g,b)=b+1$.
Since there are $a+m$ tokens smaller than $a$, of which at most $m-1$ tokens can be in $b \cdot \vec{i}$,
we have $\mrm{Inv}(g,a) \ge a + 1$.  That is, $\pi(g,a)=\pi(f,a)$.
On the other hand, since there are $b+m$ tokens smaller than $b$, of which at most $m$ tokens can be in $a \cdot \vec{i}$,
we have $\mrm{Inv}(f,b) \ge b$.

\prg{Case 2.1.} $\mrm{Inv}(f,b) = b$, which means $\pi(f,b)=b = \pi(g,b)-1$.
In this case, all the elements of $\vec{i}$ must be smaller than $b$.
We have 
\begin{align*}
\delta(f) &= \delta(g)  =  \delta(a \cdot \vec{i};\, \vec{j})
\,.\end{align*}
Therefore, $\Psi(g)=\Psi(f)+1$.

\prg{Case 2.2.} $\mrm{Inv}(f,b) > b$, which means $\pi(f,b)=b+1=\pi(g,b)$.
In this case, there must be $c > b$ in $\vec{i}$. We have
\begin{align*}
\delta(f) &= \delta(a \cdot \vec{i};\, b \cdot \vec{j}) =  \delta(a \cdot \vec{i}[b/c];\, \vec{j})\,,
\\
\delta(g) &= \delta(b \cdot \vec{i};\, a \cdot \vec{j}) =  \delta(b \cdot \vec{i}[a/c];\, \vec{j})
\,.\end{align*}
Lemma~\ref{lem:swappositive} implies $|\delta(g)-\delta(f)|=1$.
Therefore, $\Psi(g)=\Psi(f)+1$.
\end{proof}

\begin{lemma}\label{lem:posneg}
Suppose that $g$ is obtained from $f$ by swapping a non-negative token and a negative one.
Then $|\Psi(g)-\Psi(f)|=1$.
\end{lemma}
\begin{proof}
{\ }
\prg{Case 1.} The swap takes place on non-negative vertices.
Let
\begin{gather*}
f = (\vec{i};\, \vec{j}_1 \cdot \lrangle{a,b}\cdot \vec{j}_2),
\\
g = (\vec{i};\, \vec{j}_1 \cdot \lrangle{b,a}\cdot \vec{j}_2),
\end{gather*}
where $a < 0 \le b$.
Obviously, $\mu(f)=\mu(g)$, $\mrm{Inv}(g,b) = \mrm{Inv}(f,b)+1$ and $\mrm{Inv}(g,k) = \mrm{Inv}(f,k)$ for any other $k \ge 0$.
By Lemma~\ref{lem:sp_shorten}, there exists $\vecp{i}$ consisting of the $m$ smallest tokens from $\vec{i}\cdot\vec{j}_1$ such that
\begin{align*}
	\delta(f) &= \delta(\vecp{i};\, \lrangle{a,b}\cdot \vec{j}_2) + \alpha,
\\	\delta(g) &= \delta(\vecp{i};\, \lrangle{b,a}\cdot \vec{j}_2) + \alpha
\end{align*}
for some $\alpha \le 0$.

\prg{Case 1.1.}
  $\mrm{Inv}(f,b) < b$ and $\mrm{Inv}(g,b) < b+1$.
In this case, we have $\pi(g) = \pi(f)+1 $.
Corollary~\ref{cor:sp_remove} applies to both $f$ and $g$ and we obtain $\delta(g) = \delta(f)$ and $\Psi(g)=\Psi(f)+1$.

\prg{Case 1.2.}
  $\mrm{Inv}(f,b) = b$ and $\mrm{Inv}(g,b) = b+1$.
In this case, we have $\pi(g) = \pi(f)+1$.
It is enough to show $\delta(f)=\delta(g)$.
Since $\vec{j}_2$ contains $b$ tokens smaller than $b$, $\vec{i}\cdot\vec{j}_1$ and $\vecp{i}$ contain exactly $m-1$ tokens smaller than $b$.
Let $c = \max(\vecp{i})$, which is the unique element of $\vecp{i}$ bigger than $b$.
Let $(\vecpp{i};\,d) = \gamma(\vecp{i};\,a)$.

Suppose $c=d$. We have $\gamma(\vecp{i}[b/d];\,a) = (\vecpp{i};\,b)$, where $b$ is the biggest in $\vecp{i}[b/d]$. Thus 
\begin{align*}
	\delta(f) &= \delta(\vecpp{i};\, {b}\cdot \vec{j}_2) - 1 + \alpha = \delta(\vecpp{i};\, \vec{j}_2) - 1+\alpha\,,
\\	\delta(g) &= \delta(\vecp{i}[b/d];\, {a}\cdot \vec{j}_2) + \alpha = \delta(\vecpp{i};\, \vec{j}_2) - 1+\alpha
\end{align*}
by Lemma~\ref{lem:mucon}.

If $c>d$, we have $\gamma(\vecp{i}[b/c];\,a) = (\vecpp{i}[b/c];\,d)$ and
\begin{align*}
	\delta(f) &= \delta(\vecpp{i};\, \lrangle{d,b}\cdot \vec{j}_2) + \alpha = \delta(\vecpp{i}[d/c];\, {b}\cdot \vec{j}_2) + \alpha =\delta(\vecpp{i}[d/c];\, \vec{j}_2) + \alpha\,,
\\	\delta(g) &= \delta(\vecp{i}[b/c];\, {a}\cdot \vec{j}_2) + \alpha = \delta(\vecpp{i}[b/c];\, d \cdot \vec{j}_2) + \alpha = \delta(\vecpp{i}[d/c];\, \vec{j}_2) + \alpha
\end{align*}
by Lemma~\ref{lem:mucon}.

\prg{Case 1.3.} $\mrm{Inv}(f,b) > b$ and $\mrm{Inv}(g,b) > b+1$.
In this case, we have $\pi(g) = \pi(f)$. It is enough to show $|\delta(g)-\delta(f)|=1$.
Since $\vec{j}_2$ contains at least $b+1$ tokens smaller than $b$, $\vec{i}\cdot\vec{j}_1$ contains at most $m-2$ tokens smaller than $b$.
Hence $\vecp{i}$ contains at least $2$ tokens bigger than $b$.
Let $c_1$ and $c_2$ be the biggest and second biggest in $\vecp{i}$, respectively.
Let $(\vecpp{i};d) = \gamma(\vecp{i};\,a)$.

If $d=c_1$, then by Lemma~\ref{lem:mucon},
\begin{align*}
	\delta(f) &= \delta(\vecpp{i};\, {b}\cdot \vec{j}_2) - 1 + \alpha \,,
\\	\delta(g) &= \delta(\vecp{i}[b/c_1];\, {a}\cdot \vec{j}_2) + \alpha =  \delta(\vecpp{i};\, {b}\cdot \vec{j}_2) + \alpha \,.
\end{align*}

If $d=c_2$, then by Lemma~\ref{lem:mucon},
\begin{align*}
	\delta(f) &= \delta(\vecpp{i};\, \lrangle{c_2,b}\cdot \vec{j}_2) + \alpha = \delta(\vecpp{i}[c_2/c_1];\, b \cdot \vec{j}_2) + \alpha = \delta(\vecpp{i}[b/c_1];\, \vec{j}_2) + \alpha \,,
\\	\delta(g) &= \delta(\vecp{i}[b/c_1];\, {a}\cdot \vec{j}_2) + \alpha =  \delta(\vecpp{i}[b/c_1];\, \vec{j}_2) - 1+ \alpha \,.
\end{align*}

If $d<c_2$, then by Lemma~\ref{lem:mucon},
\begin{align*}
	\delta(f) &= \delta(\vecpp{i};\, \lrangle{d,b}\cdot \vec{j}_2) + \alpha = \delta(\vecpp{i}[d/c_1][b/c_2];\, \vec{j}_2) + \alpha\,,
\\	\delta(g) &= \delta(\vecp{i}[b/c_1];\, {a}\cdot \vec{j}_2) + \alpha =  \delta(\vecpp{i}[b/c_1];\, d \cdot \vec{j}_2) + \alpha = \delta(\vecpp{i}[b/c_1][d/c_2];\, \vec{j}_2) + \alpha\,.
\end{align*}
Lemma~\ref{lem:swappositive} implies $|\delta(g)-\delta(f)|=|\delta(\vecpp{i}[b/c_1][d/c_2])-\delta(\vecpp{i}[d/c_1][b/c_2])|=1$.

\prg{Case 2.} 
The swap takes place on $0$ and a negative vertex.
Without loss of generality we may assume
\begin{gather*}
f = (a \cdot \vec{i};\, b \cdot \vec{j}),
\\
g = (b \cdot \vec{i};\, a \cdot \vec{j}),
\end{gather*}
where $a < 0 \le b$.
For the case where $a=-1$, we have already proved that $\Psi(g)=\Psi(f)+1$ in Lemma~\ref{lem:algorithm}.
Hereafter we assume that $a \neq -1$.
Clearly $\pi(g,i)=\pi(f,i)$ for all $i \in \{0,\dots,n\}-\{b\}$ and $\pi(g,b)=b+1$.
$\pi(f,b)=b$ if and only if every token in $\vec{i}$ is smaller than $b$.
Let $(\vecp{i};\,d) = \gamma(b \cdot \vec{i};\, a) $.

\prg{Case 2.1.} Suppose $\pi(f,b)=b$, in which case $\pi(g) = \pi(f)+1$.

If $d=b$, there is $k \ge 1$ such that $g^i(a) < -1$ for $i \in \{0,\dots,k-1\}$ and $g^k(a)=-1$.
Since $f(i)=g(i)$ for $i < -1$ and $f(-1)=a$, we have $[a]_f \in \Delta_f$.
Hence $\Delta_f = \Delta_g \cup \{[a]_f\}$ and thus $\mu(g)=\mu(f)-1$.
We have
\begin{align*}
	\delta(f) &= \delta(a \cdot \vec{i};\, b \cdot \vec{j})= \delta(a \cdot \vec{i};\, \vec{j})\, ,
\\	\delta(g) &= \delta(b \cdot \vec{i};\, {a} \cdot \vec{j}) = \delta(\vecp{i};\, \vec{j}) - 1\,.
\end{align*}
Since $\vecp{i}$ is the $a$-resolution of $a \cdot \vec{i}$, 
by Lemma~\ref{lem:resolution}, we have $\delta(\vecp{i};\, \vec{j}) =\delta(a \cdot \vec{i};\, \vec{j})$ and thus $\Psi(g)=\Psi(f)-1$.

If $d < b$, $[a]_f \notin \Delta_f$. In this case, $\mu(g)=\mu(f)$ and 
$\vecp{i}$ has the form $b \cdot \vecpp{i}$.
\begin{align*}
	\delta(f) &= \delta(a \cdot \vec{i};\, \vec{j}) \,,
\\	\delta(g) &= \delta(b \cdot \vecpp{i};\, d \cdot \vec{j})=\delta(d \cdot \vecpp{i};\, \vec{j})\,.
\end{align*}
Since $d \cdot \vecpp{i}$ is the $a$-resolution of $a \cdot \vec{i}$,
by Lemma~\ref{lem:resolution},
we have $\delta(a \cdot \vec{i};\, \vec{j}) = \delta(d \cdot \vecpp{i};\, \vec{j})$ and thus $\Psi(g)=\Psi(f)+1$.

\prg{Case 2.2.} Suppose $\pi(f,b)=b+1$, in which case $\pi(g) = \pi(f)$.
Since there are at least $b+1$ tokens in $\vec{j}$ smaller than $b$, there are at most $m-1$ tokens smaller than $b$ in $\vec{i}$.
Let $c = \max(\vec{i})$, which is therefore bigger than $b$.

If $d=c$, $[a]_f \notin \Delta_f$. In this case, $\mu(g)=\mu(f)$ and 
\begin{align*}
	\delta(f) &= \delta(a \cdot \vec{i}[b/c];\, \vec{j}) ,
\\	\delta(g) &= \delta(b \cdot \vec{i};\, {a} \cdot \vec{j}) = \delta(\vecp{i};\, \vec{j}) - 1\,.
\end{align*}
Since $\vecp{i}$ is the $a$-resolution of $a \cdot \vec{i}[b/c]$,
we have $ \delta(g) = \delta(f)-1$ and thus $\Psi(g)=\Psi(f)-1$.

If $d=b$, $[a]_f \in \Delta_f$ and $[b]_g \notin \Delta_g$ by the same reason as in Case~2.1.
In this case, $\mu(g)=\mu(f)-1$ and 
\begin{align*}
	\delta(f) &= \delta(a \cdot \vec{i};\, b \cdot\vec{j}) = \delta(a \cdot \vec{i}[b/c];\,  \vec{j}) ,
\\	\delta(g) &= \delta(b \cdot \vec{i};\, {a} \cdot \vec{j}) = \delta(\vecp{i};\, b \cdot \vec{j}) = \delta(\vecp{i}[b/c];\, \vec{j}) \,.
\end{align*}
Since $\vecp{i}[b/c]$ is the $a$-resolution of $a \cdot \vec{i}[b/c]$, 
 we have $\delta(f)=\delta(g) $ and thus $\Psi(g)=\Psi(f)-1$.

If $d \notin \{ b,c \}$, $[a]_f \notin \Delta_f$. In this case, $\mu(g)=\mu(f)$.
There is $\vecpp{i}$ such that $\vecp{i}=b \cdot \vecpp{i}$.
Let $h$ be obtained from $g$ by exchanging the tokens $b$ and $d$. Then $|\delta(h)-\delta(g)|=1$ and
\begin{align*}
	\delta(f) &= \delta(a \cdot \vec{i};\, b \cdot \vec{j}) = \delta(a \cdot \vec{i}[b/c];\,  \vec{j}) ,
\\	\delta(g) &= \delta(b \cdot \vec{i};\, {a} \cdot \vec{j}) = \delta(b\cdot\vecpp{i};\, d \cdot \vec{j}) = \delta(b\cdot\vecpp{i}[d/c];\, \vec{j})\,,
\\	\delta(h) &= \delta(d \cdot (\vec{i}[b/d]);\, {a} \cdot \vec{j}) = \delta(d\cdot\vecpp{i};\, b \cdot \vec{j}) = \delta(d\cdot\vecpp{i}[b/c];\, \vec{j})\,.
\end{align*}
Since $d\cdot\vecpp{i}[b/c]$ is the $a$-resolution of $a \cdot \vec{i}[b/c]$, we have $\delta(f)=\delta(h)$ and thus $|\delta(g)-\delta(f)|=1$. $|\Psi(g)-\Psi(f)|=1$.
\end{proof}

\begin{lemma}\label{lem:negneg}
Suppose that $g$ is obtained from $f$ by swapping negative tokens.
Then $|\Psi(g)-\Psi(f)|=1$.
\end{lemma}
\begin{proof}
Clearly $\pi(f)=\pi(g)$.

\prg{Case 1.} The swap takes place on non-negative vertices.
Clearly $\mu(f)=\mu(g)$.
It is enough to show $|\delta(g)-\delta(f)|=1$.
We may assume by Lemma~\ref{lem:sp_shorten}
\begin{align*}
\delta(f) &= \delta(\vec{i};\, \lrangle{a,b}\cdot \vec{j})\,,
\\
\delta(g) &= \delta(\vec{i};\, \lrangle{b,a}\cdot \vec{j})\,,
\end{align*}
where $a,b <0$.
Let $(\vec{i}_a;a') = \gamma(\vec{i};\,a)$ and $(\vec{i}_b;b') = \gamma(\vec{i};\,b)$.
It is easy to see that there is $\vec{i}_{a,b}$ such that $\gamma(\vec{i}_a;\,b)=(\vec{i}_{a,b};\,b')$ and $\gamma(\vec{i}_b;\,a)=(\vec{i}_{a,b};\,a')$.
Without loss of generality we assume $0 \le a' < b'$.
Let $c_1$ and $c_2$ be the biggest and the second biggest in $\vec{i}$.

\prg{Case 1.1.} $b'=c_1$.
By Lemma~\ref{lem:mucon},
\begin{align*}
	\delta(f) &=  \delta(\vec{i}_a;\, \lrangle{a',b}\cdot \vec{j})
	= \delta(\vec{i}_a[a'/b'];\, {b}\cdot \vec{j}) 
	= \delta(\vec{i}_{a,b};\, {a'}\cdot \vec{j}) - [a'=c_2]
\,,
\\
	\delta(g) &= \delta(\vec{i}_b;\, {a}\cdot \vec{j}) - 1
 = \delta(\vec{i}_{a,b};\, {a'}\cdot \vec{j}) - 1 - [a'=c_2]
\,,\end{align*}
where $[a'=c_2] = 1$ if $a'=c_2$ and $[a'=c_2]=0$ otherwise.
Therefore, $|\delta(f)-\delta(g)|=1$.

\prg{Case 1.2.} $b' = c_2$.
\begin{align*}
	 \delta(f)
	 &= \delta(\vec{i}_{a};\, \lrangle{a',b} \cdot \vec{j})
	= \delta(\vec{i}_{a}[a'/c_1];\, b \cdot \vec{j})
	= \delta(\vec{i}_{a,b}[a'/c_1];\, \vec{j}) - 1\,,
\\
	\delta(g) &= \delta(\vec{i}_b;\, \lrangle{b',a}\cdot \vec{j})
= \delta(\vec{i}_{b}[b'/c_1];\, {a}\cdot \vec{j})
 = \delta(\vec{i}_{a,b}[b'/c_1];\, {a'}\cdot \vec{j})
\\ &
= \delta(\vec{i}_{a,b}[a'/c_1];\, \vec{j})
\,.\end{align*}

\prg{Case 1.3.} $b'<  c_2$.
\begin{align*}
 \delta(f)
 &= \delta(\vec{i}_{a,b}[a'/c_1][b'/c_2];\, \vec{j})\,,
\\
	\delta(g) &= \delta(\vec{i}_b;\, \lrangle{b',a}\cdot \vec{j})
 = \delta(\vec{i}_{b}[b'/c_1];\, {a}\cdot \vec{j})
 = \delta(\vec{i}_{a,b}[b'/c_1];\, {a'}\cdot \vec{j})
\\ &
= \delta(\vec{i}_{a,b}[b'/c_1][a'/c_2];\, \vec{j})
\,.\end{align*}
By Lemma~\ref{lem:swappositive}, $|\delta(g)-\delta(f)|=1$.

\prg{Case 2.} The swap takes place between a negative vertex and $0$.
Without loss of generality we may assume
\begin{gather*}
f = (a \cdot \vec{i};\, b \cdot \vec{j}),
\\
g = (b \cdot \vec{i};\, a \cdot \vec{j}),
\end{gather*}
where $a,b < 0$.
For the case where $a=-1$ or $b=-1$, we have already proved that $|\Psi(g) - \Psi(f)|=1$ in Lemma~\ref{lem:algorithm}.
So we assume $a,b \neq -1$.
Let $ (\vec{i}_b;\,b') = \gamma(a \cdot \vec{i};\,b)$.
There are negative tokens $b_0,\dots,b_k < 0$ in $a \cdot \vec{i}$ such that
$b_i=f^i(b) < 0$ for all $i \le k$ and $f(b_{k}) = b' \ge 0$.
Similarly for $ (\vec{i}_a;\,a') = \gamma(b \cdot \vec{i};\,a)$,
there are negative tokens $a_0,\dots,a_l < 0$ in $b \cdot \vec{i}$ such that
$a_i=g^i(a) < 0$ for all $i \le l$ and $g(a_{l}) = a' \ge 0$.
Let $\theta_a$ and $\theta_b$ be replacements $[a_0/a_1,\dots,a_{l-1}/a_l,a_l/a']$ and $[b_0/b_1,\dots,b_{k-1}/b_k,b_k/b']$, respectively.
Then $\vec{i}_a = (b \cdot \vec{i})\theta_a$ and $\vec{i}_b = (a \cdot \vec{i})\theta_b$.

\prg{Case 2.1.}
Suppose that the sequence $\lrangle{a_0,\dots,a_l}$ contains $-1$.
By $g(-1)=b$, we have
\[
	\lrangle{a_0,\dots,a_l,a'}=\lrangle{a_0,\dots,a_{l-k-2},-1,b_0,\dots,b_k,b'}\,,
\]
where $a_{l-k-1}=-1$, $a_{l-k} = b_0$ and $a'=b'$.
Since $f^{l-k-1}(a)=g^{l-k-1}(a)=-1$ and $f(-1)=a$, we have $[a]_f = \{a_0,\dots,a_{l-k-1}\} \in \Delta_f$.
On the other hand, $g^{k+1}(b)=f^{k+1}(b)=b' \ge 0$ means that $[b]_f \notin \Delta_f$ and $[b]_g=[a]_g \notin \Delta_g$.
Therefore, $\mu(f) = \mu(g) + 1$.
Observing that
\begin{align*}
	\vec{i}_b &= (a \cdot \vec{i})\theta_b\,,
\\	\vec{i}_a &= (b \cdot \vec{i})[a_0/a_1,\dots,a_{l-k-1}/a_{l-k}]\theta_b
\\ &= (a \cdot \vec{i})[b_0/a_0][a_0/a_1,\dots,a_{l-k-1}/a_{l-k}]\theta_b
\\ &= (a \cdot \vec{i})[a_{l-k-1}/a_0,a_0/a_1,\dots,a_{l-k-2}/a_{l-k-1}]\theta_b
\\	&= (a \cdot \vec{i})\theta_b[a_0/a_1,\dots,a_{l-k-2}/a_{l-k-1},a_{l-k-1}/a_0]
\\	&= \vec{i}_b [a_0/a_1,\dots,a_{l-k-2}/a_{l-k-1},a_{l-k-1}/a_0]\,,
\end{align*}
we see that $\vec{i}_a$ is the $a_0$-resolution of $\vec{i}_b$.
Therefore, by Lemmas~\ref{lem:mucon} and~\ref{lem:resolution},
\begin{align*}
\delta(f) &= \delta(\vec{i}_b;\, {b'}\cdot \vec{j}) - d = \delta(\vec{i}_a;\, {b'}\cdot \vec{j}) - d = \delta(g)
\,,\end{align*}
where $d=1$ if $b'=\max(\vec{i})$ and $d=0$ otherwise.
All in all, we have $|\Psi(f)-\Psi(g)|=1$.

The case where $\lrangle{b_0,\dots,b_k}$ contains $-1$ can be treated in the same way.

\prg{Case 2.2.} Suppose that $-1$ occurs neither in $\lrangle{a_0,\dots,a_l}$ nor $\lrangle{b_0,\dots,b_k}$.
It is easy to see that the two sequences $\lrangle{b_0,\dots,b_k,b'}$ and $\lrangle{a_0,\dots,a_l,a'}$ have no common elements.
Hence $[a_0]_f \notin \Delta_f$ and $[b_0]_g \notin \Delta_g$.
We obtain $\mu(f)=\mu(g)$.
Without loss of generality, we may assume $a' < b'$.
Let $h = f[a'/b',b'/a']$ be obtained from $f$ by exchanging the positions of the tokens $a'$ and $b'$.
Since Lemma~\ref{lem:swappositive} ensures $|\delta(f)-\delta(h)|=1$,
it is enough to show $\delta(g)=\delta(h)$.
By Lemma~\ref{lem:mucon} and the fact $a' < b' \le \max(\vec{i})$,
\begin{align*}
\delta(g) &= \delta(\vec{i}_a;\, a' \cdot \vec{j}) = \delta((b_0 \cdot \vec{i})\theta_a;\, a' \cdot \vec{j})
\end{align*}
The $b_0$-resolution of $\vec{i}_a$ is 
\begin{align*}
	&  (b_0 \cdot \vec{i})\theta_a [b_{0}/b_{1},\dots,b_{k-1}/b_k,b_k/b',b'/b_0]
	 = (b' \cdot \vec{i})\theta_a\theta_b
\,.\end{align*}
On the other hand,
\begin{align*}
\delta(h) &= \delta(f[a'/b',b'/a'])
=  \delta((a_0 \cdot \vec{i})[a'/b',b'/a'][b_0/b_1,\dots,b_{k-1}/b_k,b_k/a'];\, a' \cdot \vec{j})
 \\& =  \delta((a_0 \cdot \vec{i})[b_0/b_1,\dots,b_{k-1}/b_k,b_k/b',b'/a'];\, a' \cdot \vec{j}) 
  =  \delta((a_0 \cdot \vec{i})\theta_b[b'/a'];\, a' \cdot \vec{j}) \,.
\end{align*}
 The $a_0$-resolution of $(a_0 \cdot \vec{i})\theta_b[b'/a']$ is given as 
\begin{align*}
	& (a_0 \cdot \vec{i})\theta_b[b'/a'][a_0/a_1,\dots,a_{l-1}/a_{l},a_l/b',b'/a_0]
	\\ &= (a_0 \cdot \vec{i})\theta_b[a_l/a'][a_0/a_1,\dots,a_{l-1}/a_{l},b'/a_0]
	\\ &= (b' \cdot \vec{i})\theta_b [a_0/a_1,\dots,a_{l-1}/a_{l},a_l/a']
	\\ &= (b' \cdot \vec{i})\theta_b \theta_a= (b' \cdot \vec{i}) \theta_a\theta_b
\,,\end{align*}
since $\theta_a$ and $\theta_b$ are independent.
Therefore,
\(\delta(g) =\delta(h) 
\) by Lemma~\ref{lem:resolution}.
\end{proof}

\begin{theorem}\label{thm:gerbera}
Token swapping on star-path graphs can be solved in polynomial time.
\end{theorem}
\begin{proof}
By Lemmas~\ref{lem:pospos}, \ref{lem:posneg}, \ref{lem:negneg} and~\ref{lem:algorithm}, the number of swaps needed is exactly $\Psi(f)$.
Obviously $\Psi$ is computable in polynomial time.
\end{proof}

\section{Proof that the PPN-Separable 3SAT Is NP-hard}\label{sec:SAT}
	We show the NP-hardness of \SAT{} by a reduction from the (usual) 3SAT~\cite{Cook71}.
	For a given CNF ${F} $ on $X$, we may without loss of generality assume that for each $x \in X$,
	the positive literal $x$ and the negative one $\neg x$ occur exactly the same number of times in ${F}$.
	Otherwise, if $x$ occurs $k$ more times than $\neg x$ does,
	we add clauses $\{\neg x,y_i,\neg y_i\}$ to ${F}$ for all $i \in \{1,\dots, k\}$ where $y_i$ are new Boolean variables.
	Now, for a given CNF ${F}$ on $X = \{ x_1 ,\dots, x_m \}$
	such that the positive and negative literals $x_i$ and $\neg x_i$ occur exactly the same number of times for each Boolean variable $x_i \in X$,
	we construct ${F}'={F}_1 \cup {F}_2 \cup {F}_3$ on $X'$ such that
	\begin{itemize}
	\item ${F} $ is satisfiable if and only if ${F} '$ is satisfiable,
	\item each positive literal $x_i$ occurs just once in each of ${F}_1$ and ${F}_2$,
	\item each negative literal $\neg x_i$ occurs just once in ${F}_3$.
	 \end{itemize}
	Let $n_i$ be the number of occurrences of the positive literal $x_i$ in ${F}$
	(thus of the negative literal $\neg x_i$)
	for each $x_i \in X$.
	\begin{enumerate}
		\item Let $X' =
			\{\, {x}_{i,j} , \bar{x}_{i,j} \mid 1 \le i \le m,\, 1 \le j \le n_i \,\}$.
		\item Let ${F}_1$ be obtained from ${F}$ by replacing the $j$-th occurrence of the positive literal ${x}_i$
		 with ${x}_{i,j}$,
		and the $j$-th occurrence of the negative literal $\neg {x}_i$ with $\bar{x}_{i,j}$ for $ j \in \{1,\dots, n_i\}$.
		\item Let ${F}_2 = \{\, \{x_{i,j}, \bar{x}_{i,j} \} \mid 1 \le i \le m \text{ and } 1 \le j \le n_i\,\}$.
		\item Let ${F}_3 = \{\, \{ \neg x_{i,j}, \neg \bar{x}_{i,j+1} \} \mid 1 \le i \le m \text{ and } 1 \le j < n_i\,\} \cup
			 \{\, \{ \neg {x}_{i,n_i}, \neg \bar{x}_{i,1} \} \mid 1 \le i \le m \,\}$.
	\end{enumerate}
	Clearly ${F}'$ is an instance of \SAT{}.
	If a map $\phi: X \to \{0,1\}$ satisfies $F$, then $\phi': X' \to \{0,1\}$ satisfies $F'$ where $\phi'(x_{i,j})=1-\phi'(\bar{x}_{i,j}) =\phi(x_i)$ for each $i$ and $j$.
	Conversely, suppose that $F'$ is satisfied by $\phi': X' \to \{0,1\}$.
	The clauses of $F_2$ and $F_3$ ensure that $\phi'(x_{i,j})=1-\phi'(\bar{x}_{i,j})=\phi'(x_{i,1})$ for all $j \in \{1,\dots,n_i\}$.
	Then it is now clear that $\phi$ defined by $\phi(x_{i}) = \phi'(x_{i,1})$ satisfies $F$.

\end{document}